\documentclass[a4paper,11pt]{article}
\usepackage[utf8]{inputenc}
\usepackage{url}
\usepackage{amssymb}
\usepackage{amsmath}
\usepackage{amsthm}
\usepackage{graphicx}
\usepackage{caption}
\usepackage{subcaption}
\usepackage[usenames,dvipsnames]{color}
\usepackage{color}
\usepackage{lineno}
\usepackage{pdfpages}

\newtheorem{theorem}{Theorem}
\newtheorem{corollary}{Corollary}

\newtheorem{lemma}{Lemma}

\newtheorem{property}{Property}

\newcommand{\G}{\mathcal{G}_{2}(S)}
\newcommand{\ch}{{\rm ch}}
\newcommand{\lca}{{\rm lca}}

\newcommand{\cupdot}{\mathbin{\mathaccent\cdot\cup}}

\graphicspath{{images/}}

\date{}

\begin{document}

\title{Geometric Biplane Graphs II: Graph Augmentation
\thanks{A preliminary version of this paper has been presented at the
\emph{Mexican Conference on Discrete Mathematics and Computational Geometry}, Oaxaca, M\'{e}xico, November 2013.}}

\author{
Alfredo Garc\'{i}a$^1$\and
Ferran Hurtado$^2$\and
Matias Korman$^{3,4}$\and
In\^{e}s Matos$^5$\and
Maria Saumell$^6$\and
Rodrigo I. Silveira$^{5,2}$\and
Javier Tejel$^1$\and
Csaba D. T\'{o}th$^7$
}

\maketitle

\footnotetext[1]{Departamento de M\'{e}todos Estad\'{\i}sticos, IUMA, Universidad de Zaragoza, Zaragoza, Spain.\\
\url{olaverri@unizar.es, jtejel@unizar.es}.}
\footnotetext[2]{Departament de Matem\`atica Aplicada II, Universitat Polit\`{e}cnica de Catalunya, Barcelona, Spain.\\
\url{ferran.hurtado@upc.edu, rodrigo.silveira@upc.edu}.}
\footnotetext[3]{National Institute of Informatics (NII), Tokyo, Japan. \url{korman@nii.ac.jp}}
\footnotetext[4]{Kawarabayashi Large Graph Project, ERATO, Japan Science and Technology Agency (JST).}
\footnotetext[5]{Departamento de Matem\'atica \& CIDMA, Universidade de Aveiro, Aveiro, Portugal,  \url{ipmatos@ua.pt, rodrigo.silveira@ua.pt}.}
\footnotetext[6]{Department of Mathematics and European Centre of Excellence NTIS (New Technologies for the Information Society), University of West Bohemia, Pilsen, Czech Republic. \url{saumell@kma.zcu.cz}.}
\footnotetext[7]{Department of Mathematics, California State University Northridge, Los Angeles, USA. \url{cdtoth@acm.org}.}

\begin{abstract}
We study biplane graphs drawn on a finite point set $S$ in the plane in general position. This is the family of geometric graphs whose vertex set is $S$ and which can be decomposed into two plane graphs. We show that every sufficiently large point set admits a 5-connected biplane graph and that there are arbitrarily large point sets that do not admit any 6-connected biplane graph. Furthermore,
we show that every plane graph (other than a wheel or a fan) can be augmented into a 4-connected biplane graph. However, there are arbitrarily large plane graphs that cannot be augmented to a 5-connected biplane graph by adding pairwise noncrossing edges.
\end{abstract}




\section{Introduction}

In a \emph{geometric graph} $G=(V,E)$, the vertices are distinct points in the plane in general position (that is, no three points in $S$ are collinear) and the edges are straight line segments between pairs of vertices. A \emph{plane graph} is a geometric graph in which no two edges cross. It is well known that every planar graph can be realized as a plane graph by F\'ary's theorem~\cite{F48}.

We consider a generalization of plane graphs. A geometric graph $G=(V,E)$ is \emph{$k$-plane} for some $k\in \mathbb{N}$ if its edge set can be partitioned into $k$ disjoint subsets, $E=E_1\cupdot\ldots \cupdot E_k$, such that $G_1=(V,E_1), \ldots , G_k=(V,E_k)$ are all plane graphs, where  $\cupdot$ represents the disjoint union. For a finite point set $S$ in the plane in general position, denote by $\mathcal{G}_k(S)$ the family of $k$-plane graphs with vertex set $S$.
With this terminology, $\mathcal{G}_1(S)$ is the family of plane graphs with vertex set $S$, and $\G$ is the family of 2-plane graphs (also known as \emph{biplane graphs}) with vertex set $S$.

We contrast several combinatorial properties of plane graphs $\mathcal{G}_1(S)$ and biplane graphs $\G$ for point sets $S$ in this and a companion paper~\cite{GHKMSSTT13-I}. For example, it is well known that the vertex connectivity of every plane graph in $\mathcal{G}_1(S)$ is at most 5, but it is natural to expect that the larger family $\G$ contains graphs of higher vertex connectivity. A graph $G=(V,E)$ in $\G$ is maximal if there is no graph $G'=(V,E')$ in $\G$ such that $E \subset E'$.
In~\cite{GHKMSSTT13-I} we compare combinatorial properties of maximal biplane graphs in $\G$ with triangulations in $\mathcal{G}_1(S)$, and show that there are arbitrarily large point sets $S$ for which $\G$ contains an 11-connected graph, but no biplane graph is 12-connected.
In this paper, we study the maximum vertex connectivity of a graph in $\G$ for a given point set $S$. We also consider closely related connectivity augmentation problems. We refer to \cite{GHKMSSTT13-I} and the references therein for a broad overview of further related work.\footnote{\textbf{Note:} The companion paper~\cite{GHKMSSTT13-I} contains a larger introduction to the concept of \emph{biplane graphs}, comparing it with other related concepts such as the \emph{geometric thickness} and others. Preliminary versions of both papers have been presented at the \emph{Mexican Conference on Discrete Mathematics and Computational Geometry (Oaxaca, 2013)}. To avoid repetition, a complete introduction is presented in the companion paper, and here we include only a brief self-contained introduction. }

\smallskip\noindent{\bf Organization.}
The problem of constructing a plane graph with the largest possible vertex- or edge-connectivity
on a given point set has received significant attention~\cite{DDGC97,GHHTV09,GHTV13}. The combinatorial aspect of the problem asks for characterizing the point sets that admit graphs of a certain connectivity, and the algorithmic aspect is to develop algorithms for computing highly connected graphs.
These are the topics we study in Section~\ref{sec:scratch}, considering biplane graphs instead of plane graphs.

A closely related family of problems is the \emph{graph augmentation}, in which one would like to add new edges, ideally as few as possible, to a given graph in such a way that some desired property is achieved. There has been extensive work on augmenting disconnected plane graphs to connected ones (see \cite{HT13} for a recent survey) or achieving good connectivity properties \cite{AGHTU08,ISTW10,AIR09,dobrev,kranakis, RW12,Csa12}. For abstract graphs, this corresponds to the classical connectivity augmentation problem in combinatorial optimization and has a rich history, as well.
In Section~\ref{sec:augmentation}, we consider several problems on augmenting plane graphs to biplane graphs with higher connectivity.
We conclude in Section~\ref{sec:conclusion} with some final remarks and open problems.


\section{Drawing Biplane Graphs from Scratch}\label{sec:scratch}

Given a set $S$ of $n$ points in general position, we would like to construct a graph $G\in\G$ with high vertex connectivity $\kappa(G)$. We determine the maximum $\kappa(G)$ that can be attained for every (sufficiently large) point set $S$. We also consider the special case that the point set $S$ is in convex position.

For comparison, we briefly review analogous results for plane graphs $\mathcal{G}_1(S)$
on a given point set $S$. Refer to~\cite{HT13} for a survey. Every set of $n\geq 3$
points in general position admits a spanning cycle (a polygonization of $S$), which is 2-connected. For points in convex position, every plane graph has a vertex of degree 2, so in this case $\kappa(G)=2$ is the best possible value over all $G\in \mathcal{G}_1(S)$. Since every planar graph has a vertex of degree not greater than 5, the vertex connectivity of every graph in $\mathcal{G}_1(S)$ is at most 5.
It is known that every set of $n\geq 4$ points not in convex position admits a 3-connected triangulation. Additionally, every set of $n\geq 6$ points whose convex hull is a triangle admits a 4-connected triangulation, provided that a certain condition is satisfied (see~\cite{DDGC97} for details). Characterizations for point sets in general position that admit 4-connected triangulations have only recently been proposed~\cite{DGR13,GHTV13,GHTV}, and no characterization is known for 5-connectivity.

\subsection{Point Sets in Convex Position}

We begin with the special case of points in convex position. It turns out that all biplane graphs on a point set $S$ in convex position are planar. Moreover, the maximum number of edges (resp., the maximum vertex connectivity) of a graph in $\G$ is the same as the maximum attained in the family of planar graphs with $|S|$ vertices.


\begin{lemma}\label{lem:planar}
Let $S$ be a set of $n$ points in the plane in convex position.
\begin{itemize}
\item[{\rm (i)}] Every graph in $\G$ is planar (as an abstract graph).
\item[{\rm (ii)}] If $G=(V,E)$ is a Hamiltonian planar (abstract) graph with $n$ vertices,
       then it has a geometric realization in $\G$.
\end{itemize}
\end{lemma}

\noindent {\bf Remark}: In~\cite{BK}, Bernhart and Kainen show two results (Lemma 2.1 and Theorem 2.5), given in terms of book thickness, that are more general than Lemma~\ref{lem:planar}. Since it is straightforward to see that Lemma~\ref{lem:planar} follows from those results, we omit its proof.

\vskip 0.3cm
Note, however, that not every planar graph can be realized as a biplane graph on a point set in convex position. If $S$ is in convex position, then the boundary of the convex hull $\ch(S)$ forms a Hamiltonian cycle in every maximal (i.e., edge-maximal) graph in $\G$. Hence every maximal graph in $\G$ is planar and Hamiltonian. However, there are maximal planar graphs (triangulations) that are not Hamiltonian. (It is NP-complete to decide whether a maximal planar graphs is Hamiltonian~\cite{Chv85,Wid82}.) These planar graphs cannot be realized as a biplane graph on a point set in convex position.

We can now characterize the maximum vertex connectivity of a graph in $\G$ when $S$ is in convex position.

\begin{theorem} \label{theo:5connconvex}
Let $S$ be a set of $n$ points in convex position.
\begin{itemize}\itemsep -1pt
\item $\G$ contains a 4-connected graph if and only if $n\geq 6$.
\item $\G$ contains a 5-connected graph if and only if $n=12$ or $n\geq 14$.
\item $\G$ contains no 6-connected graphs for any $n\in \mathbb{N}$.
\end{itemize}
\end{theorem}
\begin{proof}
It is well known that every 4-connected planar graph $G$ has a Hamiltonian cycle~\cite{T56}.
By Lemma~\ref{lem:planar}(ii), it is enough to establish the existence or nonexistence of a $k$-connected planar graph for a given $n$ for $k=4$, 5, and 6.

By Lemma~\ref{lem:planar}(i), every graph in $\G$ is planar. Every planar graph on $n\geq 3$ vertices has at most $3n-6$ edges, and the sum of vertex degrees is at most $6n-12$. In a $k$-connected graph, the degree of every vertex is at least $k$, and the sum of vertex degrees is at least $kn$. Comparing these bounds, we have $kn\leq 6n-12$ or $12/(6-k)\leq n$. It follows that no planar graph is 6-connected, every 5-connected planar graph has at least 12 vertices, and every 4-connected planar graph has at least 6 vertices.

It is easy to see that there is a 4-connected planar graph on $n$ vertices for every $n\geq 6$. Specifically, the 1-skeleton of the octahedron is 4-connected with 6 vertices; and a vertex split operation can increase the number of vertices by one while maintaining 4-connectivity and planarity. In an embedding of a 4-connected planar graph on $n$ vertices, this split operation removes an edge $(u,v)$ and adds a new vertex connected to all the vertices of the two faces adjacent to $(u,v)$.


Barnette~\cite{Bar74} and Butler~\cite{But74} independently designed algorithms for generating all 5-connected triangulations, using simple operations starting from the icosahedron (see also \cite{BMK05}). The 1-skeleton of the icosahedron is 5-connected with 12 vertices, and each operation either splits a vertex of degree 6 or higher, or simultaneously splits two adjacent vertices.
Hence there is a 5-connected planar graph for $n=12$ and for every $n\geq 14$ (but not for $n=13$).
\end{proof}

\paragraph{Remark.}
We have shown that $\G$ contains 4- and 5-connected graphs when $n\geq 6$ and $n\geq 14$, respectively. The existence proof in Theorem~\ref{theo:5connconvex} can be turned into an $O(n)$-time algorithm for constructing such biplane graphs. Here we present explicit constructions for 5-connected biplane graphs for points in convex position when $n=12$ and $n\ge 14$. The construction in Figure~\ref{pic:5convex} (left) works when $n$ is even and $n\ge 12$. It is based on two plane spanning trees, $T_1$ and $T_2$, and the edges of $\ch(S)$. Each spanning tree consists of two stars with 3 leaves each, connected by a zig-zag path. Let us assume that the points are numbered clockwise. Then, the centers of the stars of $T_1$ (the nondashed spanning tree) are placed at opposite points, say $i$ and $i+n/2$, the leaves of the first star are placed at points $i+1, i+2$ and $i+3$ and the leaves of the second star are placed at points $i+n/2+1, i+n/2+2$ and $i+n/2+3$. The zig-zag path connects the centers of the stars visiting alternatively the points $\{i+4,\ldots , i+n/2-1\}$ and the points $\{i-1, i-2, \ldots , i+n/2+4\}$. Tree $T_2$ (the dashed spanning tree) is symmetric to $T_1$, where point $i+n/2+1$ is the image of point $i$ and $i+1$ is the image of point $i+n/2$. Note that the only common edges to both trees are $(i,i+1)$ and $(i+n/2, i+n/2+1)$. The construction in Figure~\ref{pic:5convex} (right) works when $n$ is odd and $n\ge 15$. It is analogous to the previous construction, but one of the stars in $T_1$ and $T_2$ has 4 leaves instead of 3.


\begin{figure}[htb]
\centering
\includegraphics[scale=0.5]{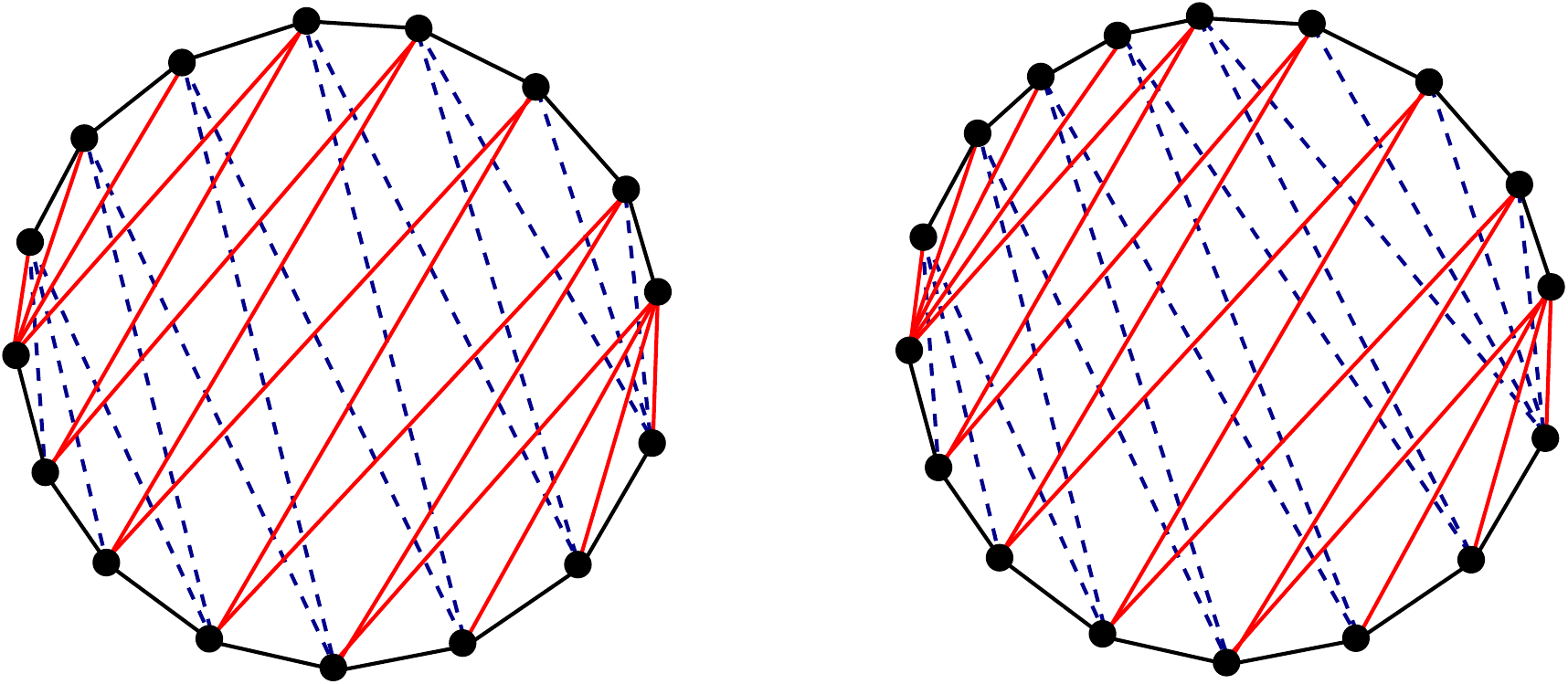}
\caption{Two 5-connected biplane graphs for points in convex position.}
\label{pic:5convex}
\end{figure}

\subsection{Point Sets in General Position}
\label{ssec:gen}

In this section we find the largest $k\in \mathbb{N}$ such that every sufficiently large point set $S$ in general position admits a $k$-connected biplane graph. Hutchinson et al.~\cite{HSV99} proved that every biplane graph in $\G$ has at most $6n-18$ edges for $n\geq 8$. In particular, this implies that every biplane graph contains a vertex of degree 11 or less, hence $k\leq 11$. Theorem~\ref{theo:5connconvex} directly improves the bound to $k\leq 5$. We now show that this bound is tight, that is, every sufficiently large point set $S$ admits a 5-connected biplane graph.

\begin{theorem} \label{theo:5conngeneric}
Let $S$ be a set of $n$ points in the plane in general position.
If $S$ contains at least 14 points in convex position,
then there is a 5-connected graph in $\G$.
\end{theorem}

Erd\H{o}s and Szekeres proved that for every $k\in\mathbb{N}$ there is an integer $f(k)$ such that every set of at least $f(k)$ points in the plane in general position contains a subset of $k$ points in convex position. They conjectured $f(k)=2^{k-2}+1$, and showed $f(k)\geq 2^{k-2}+1$. The currently best upper bound~\cite{TV05} for $k\geq 7$ is $f(k)\leq {2k-5\choose k-2}+1$. When $k=14$, this result implies that every set $S$ of $n$ points in general position has a subset of at least 14 points in convex position for $n\geq {2\cdot 14-5\choose 14-2}+1= 1352079$. Therefore, every sufficiently large point set $S$ in general position satisfies the condition in Theorem~\ref{theo:5conngeneric}.

\begin{corollary}\label{cor:5conngeneric}
If $S$ is a set of $n\geq 1352079$ points in the plane in general position, then there is a 5-connected biplane graph in $\G$.
\end{corollary}

\paragraph{Outline.}
The remainder of Section~\ref{ssec:gen} is devoted to the proof of Theorem~\ref{theo:5conngeneric}.
Our approach is as follows: given a point set $S$, let $S_0\subset S$ be a largest subset of points in convex position (in case of ties, choose a set $S_0$ whose convex hull has the largest area). By assumption, we have $|S_0|\geq 14$. By Theorem~\ref{theo:5connconvex}, $S_0$ admits a 5-connected biplane graph $G_0$. We increment $G_0$ with new vertices from $S\setminus S_0$, maintaining a 5-connected biplane graph, in 3 phases: we first insert the points lying in the interior of $\ch(S_0)$, then the vertices of $\ch(S)$, and finally all remaining points (which lie in the exterior of $\ch(S_0)$). We continue with the details.

\smallskip\paragraph{\large Preliminaries.}
To ensure that we maintain 5-connectivity, we use the following well-known properties of graphs.
\begin{property}\label{pro:property1+}
Let $G=(V,E)$ be a $k$-connected (abstract) graph. Augment $G$ with a new vertex $x$ joined to $k$ vertices of $G$. Then the new graph on vertex set $V\cup \{x\}$ is also $k$-connected.
\end{property}

\begin{property}\label{pro:property2+}
Let $G=(V,E)$ be a $k$-connected (abstract) graph in which $vw$ is an edge. Remove edge $vw$ from $G$, and augment it with a new vertex $x$ joined to both $v$, $w$ and to $k-2$ additional vertices. Then the new graph with vertex set $V\cup \{x\}$ is also $k$-connected.
\end{property}

Since adding edges to a graph can only increase the vertex connectivity, we may assume in each phase of our algorithm that
we have started with a 5-connected maximal biplane graph. Thus, we can rely on the following two structural results for maximal biplane graphs from the companion paper~\cite{GHKMSSTT13-I}.

\begin{lemma} \cite{GHKMSSTT13-I}\label{lem:tri}
Let $G=(S,E)$ be a maximal biplane graph in $\G$.
Then there are two triangulations $T_1=(S,E_1)$ and $T_2=(S,E_2)$
such that $E=E_1\cup E_2$.
\end{lemma}

Given an edge $e\in E$ in a triangulation $T=(S,E)$, we denote by $Q(e)$ the quadrilateral
formed by the two triangles adjacent to $e$. Note that $Q(e)$ is not defined if $e$ is an edge of $\ch(S)$. An edge $e$ is \emph{flippable} if and only if $Q(e)$ is a convex quadrilateral.

\begin{lemma}\cite{GHKMSSTT13-I}\label{lem:flip}
Let $G=(S,E)$ be a maximal biplane graph in $\G$ such that $E=E_1\cup E_2$, where
$T_1=(S,E_1)$ and $T_2=(S,E_2)$ are two triangulations. Every edge of $E_1\cap E_2$ is
flippable in neither $T_1$ nor $T_2$. Furthermore, every maximal biplane graph with
$n\geq 4$ vertices is 3-connected.
\end{lemma}

The following tool (Lemma~\ref{lem:flip+}) is crucial for increasing the vertex degree of
a vertex in a triangulation. This tool is applicable to all triangulations other than the wheel.
(A \emph{wheel} is a triangulation on $n$ points such that $n-1$ points are in
convex position and one point lies in the interior of $\ch(S)$, the points on the convex hull
induce a cycle on the boundary of $\ch(S)$ and the interior point is joined to all other $n-1$
points.)

\begin{lemma}\label{lem:flip+}
Let $T=(S,E)$ be a triangulation other than the wheel. Let $s\in S$ be a point in the interior of $\ch(S)$ such that it is adjacent to a vertex on the boundary of $\ch(S)$, and the graph induced by its neighbors in $T$ is a cycle. Then $T$ contains a triangle incident to $s$ in which the edge opposite to $s$ is flippable.
\end{lemma}
\begin{proof}
Denote the neighbors of $s$ by $v_1,v_2,\ldots , v_k\in S$, for some $k\geq 3$, in counterclockwise order. Since $T$ is not a wheel, some edges of the cycle $(v_1,\ldots , v_k)$ are not on the boundary of $\ch(S)$. Without loss of generality, assume that $v_1$ is a vertex of $\ch(S)$ but $v_1v_2$ is not on the boundary of $\ch(S)$. Note that $Q(v_1v_2)$ is defined, and it has a convex vertex at $v_1$.

\begin{figure}[htpb]
\centering
\includegraphics[width=5in]{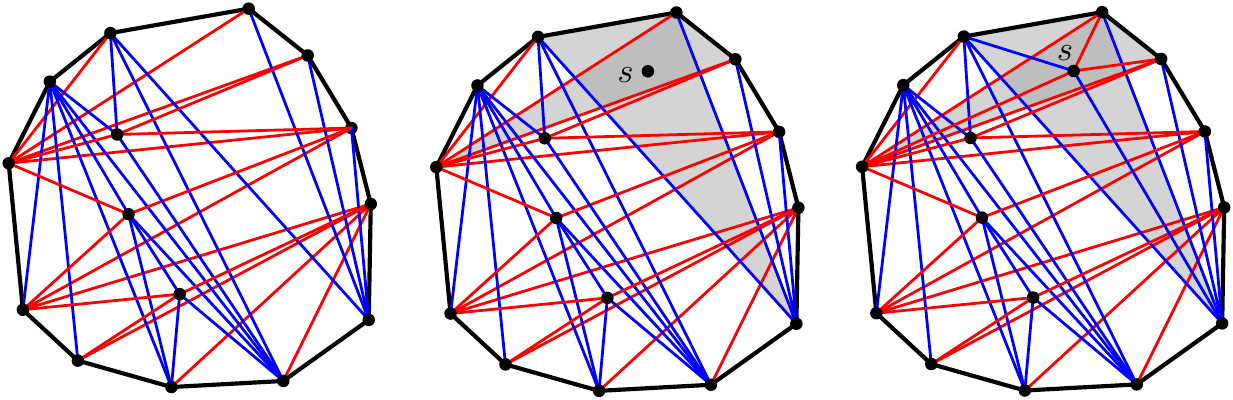}
\caption{Left: $5$-connected biplane graph on 15 points.
Middle: point $s$ lies in the interior of two gray triangles, which jointly have 5 distinct vertices.
 Right: point $s$ is now part of the 5-connected biplane graph.
}\label{pic:5ConnectedGen1}
\end{figure}

Starting with $i=1$ we use the following iterative argument: we know that $Q(v_iv_{i+1})$ is defined and has a convex vertex at $v_i$. If $v_iv_{i+1}$ is flippable, we are done. Otherwise, $Q(v_iv_{i+1})$ is a nonconvex quadrilateral and has a reflex vertex at $v_{i+1}$. It follows that $v_{i+1}$ is in the interior of $\ch(S)$, and thus $Q(v_{i+1}v_{i+2})$ is defined. Since the neighbors of $s$ induce a
cycle, we have $v_iv_{i+2}\not\in E$. Therefore $v_iv_{i+1}$ and $v_{i+1}v_{i+2}$ are not adjacent to a common triangle. Since $v_{i+1}$ is a reflex vertex of $Q(v_iv_{i+1})$, it must be a convex vertex of $Q(v_{i+1}v_{i+2})$. Thus, we can increment the value of $i$ and repeat the same argument. This process ends as soon as we find a flippable edge or when we conclude that none of the edges $v_iv_{i+1}$, $i=1,\ldots, k-1$ is flippable. However, in the latter case, the above argument implies that $v_1$ is in the interior of $\ch(S)$, contradicting our initial assumption.
\end{proof}

\smallskip\noindent{\bf Inserting interior vertices.}
The following lemma allows augmenting a 5-connected biplane graph with an \emph{interior} point.

\begin{lemma}\label{lem:incremental}
Let $G=(S,E)$ be a 5-connected biplane graph.
Denote by $S_{\rm int}\subset S$ the points lying in the interior of $\ch(S)$, and let $s\not\in  S$ be a point such that $s$ is in the interior of $\ch(S)$ but in the exterior of $\ch(S_{\rm int})$. Then a 5-connected biplane graph on $S\cup\{s\}$ can be constructed from $G=(S,E)$ by adding at least 5 new edges incident to $s$ and deleting at most one edge of $E$.
\end{lemma}
\begin{proof}
Augment $G=(S,E)$ to a 5-connected maximal biplane graph $\widehat{G}=(S,\widehat{E})$ by adding dummy edges, if necessary. By Lemma~\ref{lem:tri}, $\widehat{G}$ is the union of two triangulations, $T_1$ and $T_2$. Point $s$ lies in the interior of some triangles $\Delta_1$ and $\Delta_2$ in the two triangulations ($\Delta_1$ and $\Delta_2$ may share vertices and edges). Since $s$ lies in the exterior of $\ch(S_{\rm int})$, at least one vertex of $\Delta_1$ (resp., $\Delta_2$) is on the boundary of $\ch(S)$. Let us augment $T_1$ (resp., $T_2$) with vertex $s$ and three edges joining $s$ to the vertices of $\Delta_1$ (resp., $\Delta_2$) to a new triangulation $T_1'$ (resp., $T_2'$) in which $s$ is adjacent to a vertex of $\ch(S)$. We distinguish three cases based on the total number of distinct vertices of $\Delta_1$ and $\Delta_2$.

\paragraph{Case~1: $\Delta_1$ and $\Delta_2$ jointly have 5 or 6 distinct vertices.} We have joined $s$ to at least 5 distinct vertices of $G$ (Figure~\ref{pic:5ConnectedGen1}). The union of $T_1'$ and $T_2'$ is biplane and 5-connected by Property~\ref{pro:property1+}.

\paragraph{Case~2: $\Delta_1$ and $\Delta_2$ jointly have 4 distinct vertices.}
In this case, $\Delta_1$ and $\Delta_2$ share an edge (see Figure~\ref{pic:5ConnectedGen2}), say $\Delta_1=v_1v_2v_3$ and $\Delta_2=v_1v_2v_4$. Since $s$ is in the interior of both $\Delta_1$ and $\Delta_2$, the points $v_3$ and $v_4$ are on the same side of the line $v_1v_2$. Since $v_4$ is in the exterior of $\Delta_1$ and $v_3$ is in the exterior of $\Delta_2$, the convex hull $\ch(v_1,v_2,v_3,v_4)$ is a convex quadrilateral. Without loss of generality, we may assume  $\ch(v_1,v_2,v_3,v_4)=(v_1,v_2,v_3,v_4)$ in counterclockwise order. By Lemma~\ref{lem:flip+}, $T_1'$ (resp., $T_2'$) has a flippable edge $e_1'$ (resp., $e_2'$) in a triangle opposite to $s$. If flipping edge $e_1'$ in $T_1'$ or edge $e_2'$ in $T_2'$ increases the degree of $s$ to 5, then perform the edge flip. By Property~\ref{pro:property2+}, the union of the two triangulations is a 5-connected biplane graph.

Assume now that neither flipping $e_1'$ in $T_1'$ nor $e_2'$ in $T_2'$ increases the degree of $s$ to 5. This implies that the third vertex of the triangles adjacent to $e_1'$ and $e_2'$, respectively, are $v_4$ and $v_3$. That is, the 4-cycle $(v_1,v_2,v_3,v_4)$ is part of both triangulations $T_1'$ and $T_2'$. After flipping $e_1'$ in $T_1'$, the cycle $(v_1,v_2,v_3,v_4)$ has a flippable edge $\hat{e}$ by Lemma~\ref{lem:flip+}. We can flip $\hat{e}$ in $T_1'$ to increase the degree of $s$ to 5, while retaining edge $\hat{e}$ in the other triangulation $T_2'$. By Property~\ref{pro:property2+}, the union of these two triangulations is a 5-connected biplane graph.

\begin{figure}[htpb]
\centering
\includegraphics[width=5in]{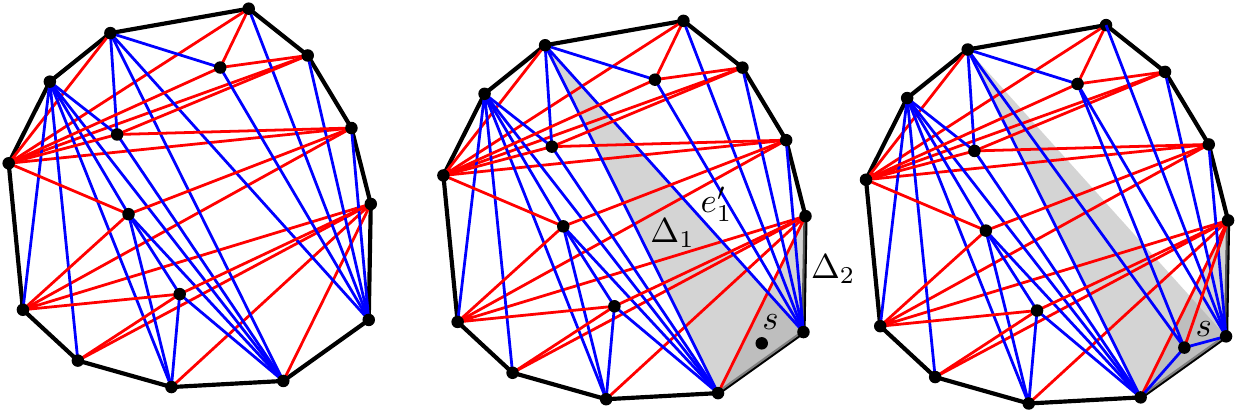}
\caption{Left: $5$-connected biplane graph on 16 points.
Middle: point $s$ lies in the interior of triangles $\Delta_1$ and $\Delta_2$, which share an edge, and $e'_1$ is a flippable edge adjacent to $\Delta_1$.
Right: point $s$ is now part of the 5-connected biplane graph. }\label{pic:5ConnectedGen2}
\end{figure}

\paragraph{Case~3: $\Delta_1$ and $\Delta_2$ jointly have 3 distinct vertices.}
In this case, $\Delta_1=\Delta_2$, say $\Delta_1=\Delta_2=v_1v_2v_3$. By Lemma~\ref{lem:flip+}, $T_1'$ (resp., $T_2'$) has a flippable edge $e_1'$ (resp., $e_2'$) in a triangle opposite to $s$. If $e_1'\neq e_2'$, then we can flip each edge in its corresponding triangulation, while keeping it in the other triangulation. Thus the degree of $s$ increased to 5, and the union of the two triangulations forms a 5-connected biplane graph by Property~\ref{pro:property1+}. If $e_1'= e_2'$ but the two flips together increase the degree of $s$ to 5, then the union of these two triangulations is a 5-connected biplane graph by Property~\ref{pro:property2+} (since the two flips together remove at most one edge $e_1'= e_2'$ from $G$).

It remains to consider the case that $e_1'=e_2'$, say $e_1'=e_2'=v_2v_3$, and the two flips together would only increase the degree of $s$ to 4. That is, $\Delta_1=\Delta_2=v_1v_2v_3$ is adjacent to the same triangle, say $v_2v_3v_4$, in both $T_1$ and $T_2$. In particular, the 4-cycle $(v_1,v_2,v_4,v_3)$ is part of both triangulations. We claim that the 4-cycle $(v_1,v_2,v_4,v_3)$ has no external chords in at least one of $T_1$ and $T_2$. Clearly, the claim is true when $(v_1,v_2,v_4,v_3)$ is a convex quadrilateral. Let us assume to the contrary that $\ch(v_1,v_2,v_3,v_4)$ is a triangle, say $\Delta = v_1v_2v_4$, and suppose that the external chord $v_1v_4$ belongs to both $T_1$ and $T_2$, so $\Delta = v_1v_2v_4$ is part of both triangulations.  If $\Delta = \ch(S)$, then the path $v_4v_3v_1$ separates $v_2$ from the rest of the vertices in $G$, and if $\Delta\neq \ch(S)$, then $v_3$ lies in the interior of $\Delta$ and some point of $S$ lies in its exterior. It follows that either the path $v_4v_3v_1$ or $\Delta$ is a 3-vertex cut in $G$, contradicting our initial assumption that $G$ is 5-connected, and proving the claim. Without loss of generality, we may now assume that the 4-cycle $(v_1,v_2,v_4,v_3)$ has no external chord in $T_1$, and hence in $T_1'$ either.


After flipping $e_1'=v_2v_3$ in $T_1'$, the 4-cycle $(v_1,v_2,v_3,v_4)$ is an induced subgraph in the resulting triangulation $T_1''$, and this cycle has a flippable edge $\hat{e}$ by Lemma~\ref{lem:flip+}. Flipping $\hat{e}$ in $T_1''$ increases the degree of $s$ to 5, while the edge $\hat{e}$ remains part of the triangulation $T_2'$. By Property~\ref{pro:property1+}, the union of these two triangulations is a 5-connected biplane graph. This completes the proof in case~3.

\paragraph{}In all three cases, we have augmented $\widehat{G}=(S,\widehat{E})$ with a new vertex $s$ by adding at least 5 new edges incident to $s$ and deleting at most one edge of $\widehat{E}$. Finally, delete all remaining dummy edges (that have not been flipped in the above procedure). Since the original graph $G$ was 5-connected without the dummy edges, Properties~\ref{pro:property1+} and \ref{pro:property2+} imply that the resulting biplane graph on $S\cup\{s\}$ is also 5-connected.
\end{proof}

\smallskip\paragraph{\large Inserting vertices of the convex hull.}
We now introduce a method to augment a 5-connected biplane graph $G=(S_a,E)$ with a set $S_b$ of points in the exterior of $\ch(S_a)$. We would like the new edges to be disjoint from the interior of $\ch(S_a)$ (although some edge flips will be necessary). For this purpose, we introduce the concept
of \emph{visibility}. We say that a point $s$ in the exterior of $\ch(S_a)$ \emph{sees} an edge $uv$ of $\ch(S_a)$ if the triangle $suv$ is also in the exterior of $\ch(S_a)$. A line segment $st$ in the exterior of $\ch(S_a)$ \emph{sees} $uv$ if both $s$ and $t$ see $uv$. Note that every exterior point $s$ must see a subset of consecutive edges of $\ch(S_a)$, but cannot see all edges of $\ch(S_a)$.

We show below (Lemma~\ref{lem:incrementalBoundary}) that a 5-connected biplane graph $G=(S_a,E)$ can be augmented with a set $S_b$ of exterior points if $S_a$ and $S_b$ satisfy the following property
(see Fig.~\ref{pic:5ConnectedBoundary}).

\begin{property}\label{pro:maxi}
Let $S_a$ and $S_b$ be disjoint point sets such that
\begin{itemize}\itemsep -2pt
\item $|\ch(S_a)|\geq 4$;
\item every point $s\in S_b$ is a vertex of $\ch(S_a\cup S_b)$;
\item if $k$ points in $S_b$ are consecutive vertices of $\ch(S_a\cup S_b)$, then they jointly see at least $k+2$ consecutive edges of $\ch(S_a)$, for all positive integers $k<|\ch(S_a)|$.
\end{itemize}
\end{property}

Property~\ref{pro:maxi} implies a similar property for edges (rather than vertices) under some additional conditions. This will allow the application of Hall's theorem to match the edges of $\ch(S_b)$
to some edges of $\ch(S_a)$.

\begin{lemma}\label{lem:maxi-edge}
Assume that $S_a$ and $S_b$ satisfy Property~\ref{pro:maxi}, $S_a$ lies in the interior of $\ch(S_b)$, and every two consecutive vertices of $\ch(S_b)$ see some common edge of $\ch(S_a)$. Then every $k$ edges of $\ch(S_b)$ jointly see at least $k$ edges of $\ch(S_a)$.
\end{lemma}
\begin{proof}
When $k=1$, the claim holds by our assumption that every two consecutive vertices of $\ch(S_b)$ see some common edge of $\ch(S_a)$. Suppose, to the contrary, that there is a counterexample for some $k>1$. That is, there is a set $H_b$ of $k\geq 2$ edges of $\ch(S_b)$ that jointly only see a set $H_a$ of edges of $\ch(S_a)$ with $|H_a|<k$. Consider a counterexample where $|H_a|$ is minimal. We may assume that $H_b$ is the maximal set of edges of $\ch(S_b)$ that jointly see exactly the edges in $H_a$. If two edges $h_1,h_2\in H_b$ see the same $h\in H_a$, then every edge along $\ch(S_b)$ between $h_1$ and $h_2$ (say, in counterclockwise order) can see only edges of $\ch(S_a)$ that are already visible to $h_1$ or $h_2$. Thus every edge in $H_a$ is visible from a sequence of \emph{consecutive} edges in $H_b$. It is clear that every edge $h\in H_b$ sees a set of consecutive edges of $\ch(S_a)$. Consequently, we may assume that both $H_b$ and $H_a$ consist of consecutive edges (along $\ch(S_b)$ and $\ch(S_a)$, respectively). The $k$ consecutive edges in $H_b$ form a path $P$. By Property~\ref{pro:maxi}, the $k-1$ interior vertices of this path jointly see at least $k+1$ consecutive edges of $\ch(S_a)$. These edges of $\ch(S_a)$ are each visible by at least two (consecutive) vertices of the path $P$: either by two interior vertices or by one interior vertex and an endpoint of $P$. Hence $|H_a|\geq k+1$, contradicting our initial assumption $|H_a|<k$.
\end{proof}

Using the preceding observations, we present our main tool for augmenting a
5-connected biplane graph with exterior points.

\begin{lemma}\label{lem:incrementalBoundary}
Let $S_a$ and $S_b$ be two point sets satisfying Property~\ref{pro:maxi}, and let $G=(S_a,E)$ be a 5-connected biplane graph in $\mathcal{G}_2(S_a)$. Then there exists a 5-connected biplane graph $G'=(S_a\cup S_b,E')$ such that $E\subset E'$.
\end{lemma}
\begin{proof}
We may assume, by adding dummy edges if necessary, that $G$ is a 5-connected maximal biplane graph.
We consider two cases depending on whether the conditions of Lemma~\ref{lem:maxi-edge} are satisfied or not.

\begin{figure}[!htb]
\centering
\includegraphics[width=5in]{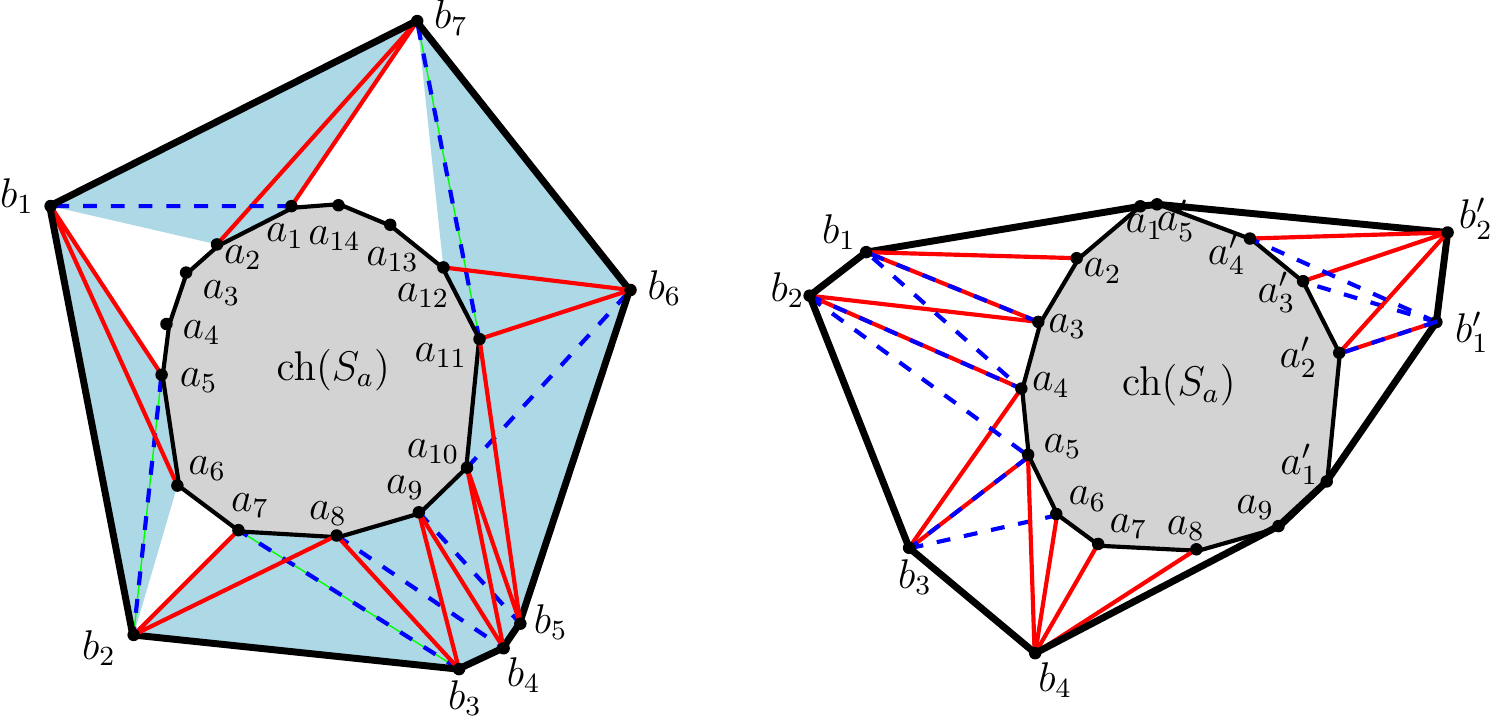}
\caption{Left: the boundaries of $\ch(S_a)$ and $\ch(S_b)$ are disjoint.
Right: the points in $S_b$ are partitioned into two treatable chains.}
 \label{pic:5ConnectedBoundary}
\end{figure}

\paragraph{Case~1: $S_a$ lies in the interior of $\ch(S_b)$ and every two consecutive vertices of $\ch(S_b)$ see some common edge of $\ch(S_a)$.} Denote the vertices of $\ch(S_a)$ by $a_1,\ldots, a_p$ in counterclockwise order, and let $b_1,\ldots ,b_q$ denote the vertices of $\ch(S_b)$ in counterclockwise order. By Lemma~\ref{lem:maxi-edge}, every set of $k$ edges of $\ch(S_b)$ jointly see at least $k$ edges of $\ch(S_a)$. Thus, using Hall's theorem, we can assign every edge of $\ch(S_b)$ to a unique visible edge of $\ch(S_a)$. We have $q\geq 3$, and Property~\ref{pro:maxi} yields $p\geq 5$.
Given an index $i \in \{1,\ldots, q\}$, let $j$ be the index such that the edge $b_ib_{i+1}$ is assigned to $a_ja_{j+1}$. By hypothesis, the quadrilateral $b_ib_{i+1}a_{j+1}a_j$ must be convex. We look for an assignment in which these quadrilaterals have pairwise disjoint interiors. If two such quadrilaterals, say $b_ib_{i+1}a_{j+1}a_j$ and $b_{i'}b_{i'+1}a_{j'+1}a_{j'}$, cross (have intersecting interiors), then $b_ib_{i+1}$ also sees $a_{j'+1}a_{j'}$ and $b_{i'}b_{i'+1}$ also sees $a_{j+1}a_j$, thus we can exchange the edges assigned to $b_ib_{i+1}$ and $b_{i'}b_{i'+1}$, reducing the total number of crossing quadrilaterals by at least one. We can now assume that the edges of $\ch(S_b)$ and the assigned edges of $\ch(S_a)$ form interior-disjoint convex quadrilaterals.

We describe how to augment $G$ with the vertices $b_1,\ldots ,b_q$. In one layer, add all edges of the cycle $(b_1,\ldots ,b_q)$. If edge $b_ib_{i+1}$ is assigned to $a_ja_{j+1}$, then join $b_i$ to $a_j$ and $a_{j+1}$ in one layer, and $b_{i+1}$ to $a_j$ in the other layer (where $a_{p+1}=a_1$ and $b_{q+1}=b_1$). Denote the resulting graph by $G'$ (Figure~\ref{pic:5ConnectedBoundary}, left). All new edges are disjoint from the interior of $\ch(S_a)$, and the edges in each layer are noncrossing, thus $G'$ is biplane. Each $b_i$ is joined to at least three vertices of the cycle $(a_1,\ldots , a_p)$, which is part of the 5-connected graph $G$, and to its two neighbors in the cycle $(b_1,\ldots , b_q)$. In particular, each $b_i$ has vertex-independent paths to five distinct vertices of the 5-connected graph $G$. It follows that $G'$ is 5-connected.

\paragraph{Case~2: $S_a$ does not lie in the interior of $\ch(S_b)$ or two consecutive vertices of $\ch(S_b)$ see disjoint sets of edges of $\ch(S_a)$.} In this case, we partition the vertices $S_b$ into maximal chains of consecutive vertices along $\ch(S_a\cup S_b)$ such that
every two consecutive vertices of a chain see a common edge of $\ch(S_a)$; and then successively augment $G$ with the vertices of the chains. We say that a counterclockwise chain $c=(b_1,\ldots ,b_q)$ along the boundary of $\ch(S_a\cup S_b)$ is \emph{treatable} if $b_i\in S_b$ for $i=1,\ldots, q$ and every edge of $c$ sees some edge of $\ch(S_a)$. A treatable chain $c$ is \emph{maximal} if it is not contained in a longer treatable chain.

Let $c=(b_1,\ldots, b_q)$ be a maximal treatable path, and let $(a_1,\ldots, a_p)$ be the counterclockwise sequence of vertices of $\ch(S_a)$ jointly visible from $c$. Property~\ref{pro:maxi} implies $p\geq q+3$ (in particular, $p\geq 4$, since $b_1$ alone sees at least 3 edges, hence at least 4 vertices of $\ch(S_a)$). We distinguish two subcases depending on the length of $c$.

\paragraph{Case~2(a): $q\geq 2$.} To each $b_i$ ($i=1,\ldots , q$) we assign a sequence of visible edges of $(a_1,\ldots , a_p)$ such that the sequences are disjoint and cover all edges of $(a_1,\ldots , a_p)$. Assign to $b_1$ the edges $a_1a_2$, $a_2a_3$, and any subsequent edge of $(a_1,\ldots , a_p)$ that is \emph{not} visible to $b_2$. For $i=2,\ldots , q$, assign to vertex $b_i$ the counterclockwise first edge of $\ch(S_a)$ that is visible to $b_i$ and has not been assigned to any previous vertex $b_j$, $i<j$; furthermore, assign to $b_i$ any subsequent edge of $(a_1,\ldots , a_p)$ that is not visible to any subsequent vertex $b_j$, $j>i$. (See Figure~\ref{pic:5ConnectedBoundary}, right.) We have assigned at least 2 edges to $b_1$ by construction, and at least one edge to all other vertices $b_i$ ($i=2,\ldots , q$) by Property~\ref{pro:maxi}.
By the maximality of the sequence $b_1,\ldots ,b_q$, every two consecutive vertices see at least one edge of $\ch(S_a)$. Using this fact, observe that $b_i$ sees the first edge assigned to $b_{i+1}$, for $i=1, \ldots , q-1$. Note that at least two edges are assigned to $b_q$. Indeed, if vertices $b_2,\ldots , b_{q-1}$ are each assigned exactly one edge, then $b_q$ is assigned at least 2 edges since $p\geq q+3$.
Otherwise consider the last vertex $b_x$, $1<x<q$, assigned to two or more edges. Then by construction, all edges visible by $b_{x+1},\ldots ,b_q$ are assigned to these vertices. By Property~\ref{pro:maxi}, the vertices $b_{x+1},\ldots ,b_q$ jointly see at least $(q-x)+2$ edges of $\ch(S_a)$. Since each of $b_{x+1},\ldots , b_{q-1}$ is assigned to exactly one edge, $b_q$ is assigned to at least 3 edges.

We can now augment $G$ with the vertices $b_1,\ldots ,b_q$. We distinguish two cases based on the number of edges assigned to $b_q$. Assume first that 3 or more edges are assigned to $b_q$ (refer to   Figure~\ref{pic:5ConnectedBoundary}, right). Then, in one layer, add all edges of the path $(b_1,\ldots ,b_q)$ and join each $b_i$ to the endpoints of all edges assigned to $b_i$. In the second layer, join every vertex $b_i$ ($i=1,\ldots , q-1$) to both endpoints of the first edge assigned to $b_{i+1}$. Denote the resulting graph by $G'$. All new edges are disjoint from the interior of $\ch(S_a)$, and the edges in each layer are noncrossing, and so $G'$ is biplane. Vertices $b_1$ and $b_q$ are each joined to at least 4 consecutive vertices of $(a_1,\ldots , a_p)$, which is part of the 5-connected graph $G$, and $b_1$ and $b_q$ are also joined by the path $(b_1,\ldots , b_q)$. Vertices $b_i$ ($i=2,\ldots , q-1$) are joined to at least 3 vertices of the arc $(a_1,\ldots , a_p)$ (two adjacent vertices in the first layer, and one additional vertex in the second layer) and they each have two neighbors in the path $(b_1,\ldots , b_q)$. It follows that $G'$ is 5-connected.

Assume now that exactly two edges are assigned to $b_q$, namely $a_{p-2}a_{p-1}$ and $a_{p-1}a_p$ (refer to Figure~\ref{pic:5ConnectedBoundary}, right). Thus, the edge $a_{p-3}a_{p-2}$ is assigned to $b_{q-1}$. Now, in one layer, add all edges of the path $(b_1,\ldots ,b_q)$, join each $b_i$ ($i=1,\ldots , q-2$) to the endpoints of all edges assigned to $b_i$, join $b_{q-1}$ to $a_{p-3}$ and join $b_q$ to the endpoints of edges $a_{p-3}a_{p-2}$, $a_{p-2}a_{p-1}$ and $a_{p-1}a_p$. In the second layer, join every vertex $b_i$ ($i=1,\ldots , q-1$) to both endpoints of the first edge assigned to $b_{i+1}$. Similarly to the previous case, the resulting graph is biplane and 5-connected.

\paragraph{Case~2(b): $q=1$.}
If $p\geq 5$, then we augment $G$ with vertex $b_1$ and join it to all visible vertices of $\ch(S_a)$. The resulting graph is biplane, and 5-connected by Property~\ref{pro:property1+}.

Otherwise, $p=4$. We augment $G$ with vertex $b_1$ and join it to $a_1,a_2,a_3,a_4$ and a fifth vertex of $G$ as follows. Since $G$ is 5-connected, $a_2$ is incident to at least 5 edges in $G$, at least 3 of which are interior edges of $\ch(S_a)$. By Lemma~\ref{lem:tri}, $G$ is the union of two triangulations, $T'$ and $T''$. We may assume without loss of generality that $a_2$ is incident to at least 2 interior edges in $T'$. Denote the triangles of $T'$ adjacent to $a_1a_2$, $a_2a_3$ and $a_3a_4$, by $\Delta_1$, $\Delta_2$ and $\Delta_3$, respectively (some of these triangles may be the same). Clearly, $\Delta_1\neq \Delta_3$, because the two triangles involve at least four different vertices, and $\Delta_1\neq \Delta_2$, because $a_2$ is incident to some interior edges. We consider two cases depending on whether $\Delta_2$ and $\Delta_3$ are equal or not. If $\Delta_1$, $\Delta_2$, and $\Delta_3$ are pairwise distinct, then one of them forms a convex quadrilateral with $b_1$, since $a_1b_1$ and $a_4b_1$ are tangent to $\ch(S_a)$. If, say, $\Delta_i=a_ia_{i+1}v_i$, forms a convex quadrilateral with $b_1$, we can join $b_1$ to $v_i$ and the resulting graph is biplane and 5-connected by Property~\ref{pro:property1+}.

Suppose that $\Delta_2=\Delta_3$, that is, $a_2a_3a_4=\Delta_2=\Delta_3$ is a triangle in $T'$. Denote the other triangle adjacent to $a_2a_4$ by $\Delta'=a_2a_4v'$. Since $a_2$ is incident to at least 2 interior edges, we have $\Delta_1\neq \Delta'$. Again, either $\Delta_1$ or $\Delta'$ forms a convex quadrilateral with $b_1$. If $\Delta_1=a_1a_2v_1$ forms a convex quadrilateral with $b_1$, then we join $b_1$ to $v_1$ and if $\Delta'$ forms a convex quadrilateral with $b_1$, then we delete edge $a_2a_4$ and we add edge $b_1v'$. The resulting graph is biplane and 5-connected by Properties~\ref{pro:property1+} or~\ref{pro:property2+}.
\end{proof}

\smallskip\paragraph{\large Proof of Theorem~\ref{theo:5conngeneric}.}

\begin{proof}
Let $S$ be a set of $n$ points in general position, with at least 14 points in convex position. Consider the subsets of $S$ that are in convex position and have maximum cardinality.
Among all these sets, let $S_0$ be one whose convex hull has the largest area. Clearly, we have $|S_0|\geq 14$ and there is a 5-connected graph $G_0\in \mathcal{G}_2(S_0)$ by Theorem~\ref{theo:5connconvex}.

Partition the points of $S\setminus S_0$ into three sets $S_{\rm int}$, $S_{\rm bou}$, and $S_{\rm ext}$ containing the points in the interior of $\ch(S_0)$, those on the boundary of $\ch(S)$, and the remaining points, respectively.
The choice of $S_0$ ensures that $S_a=S_0\cup S_{\rm int}$ and $S_b=S_{\rm bou}$ satisfy Property~\ref{pro:maxi}.

We increment $G_0$ with the vertices in $S_{\rm int}$, $S_{\rm bou}$, and $S_{\rm ext}$ in this order. First, consider the interior points in $S_{\rm int}$. Order the points in $S_{\rm int}$ in lexicographic order (i.e., sort by $x$-coordinate, and break ties by $y$-coordinates). We use Lemma~\ref{lem:incremental} for successively inserting the points in $S_{\rm int}$ into $G_0$,
maintaining a 5-connected biplane graph. Denote by $G_0'$ the resulting 5-connected biplane graph in $\mathcal{G}_2(S_0\cup S_{\rm int})$.

Since $S_0\cup S_{\rm int}$ and $S_{\rm bou}$ satisfy Property~\ref{pro:maxi}, we can use Lemma~\ref{lem:incrementalBoundary} to augment $G_0'$ into a 5-connected biplane graph
$G_0''$ on $S_0\cup S_{\rm int}\cup S_{\rm bou}$.
Finally, we use Lemma~\ref{lem:incremental} again for successively inserting the points in $S_{\rm ext}$ into $G_0''$ in a suitable order. Since $S_{\rm ext}$ lies in the exterior of $\ch(S_0)$, there exists at least one point in $S_{\rm ext}$ that lies on the boundary of $\ch(S_0\cup S_{\rm ext})$. Choose an arbitrary such point, and remove it from $S_{\rm ext}$. Repeat this process until we have removed all points of $S_{\rm ext}$. Let $p_1, \ldots ,p_k$ be the order in which points of $S_{\rm ext}$ are removed. By construction, point $p_i$ belongs to the convex hull of $S_0 \cup \{p_i, \ldots, p_k\}$, for all $i=1,\ldots , k$. Thus, if we successively insert them into $G_0''$ in reverse order they each satisfy the conditions of Lemma~\ref{lem:incremental}.
When all three phases of the algorithm are complete, we obtain a 5-connected biplane graph in $\G$, as claimed.
\end{proof}


\section{Graph Augmentation}\label{sec:augmentation}

In this section we study how the addition of a second layer can improve the vertex connectivity of a plane graph. Suppose we are given a plane graph $G=(S,E)$ in $\mathcal{G}_1(S)$, and wish to augment it with a set $E'$ of new edges to obtain a biplane graph $G=(S,E\cup E')\in \G$ such that the vertex connectivity $\kappa(G)$ is maximal. The more edges $G$ has, the more constrained the problem becomes. In the worst case, we may assume that $G\in \mathcal{G}_1(S)$ is a triangulation, and it is augmented with another triangulation.

In Section~\ref{sec:aug_incr_con} we show how to augment a triangulation with a plane graph in a second layer in order to make it 4-connected, whenever that is possible. In Section~\ref{sec:aug_minimal} we consider the \emph{minimal augmentation} problem, that is,
finding a minimum set of new edges that augments a given plane graph into a $k$-connected biplane graph.

\subsection{Increasing the connectivity of triangulations}
\label{sec:aug_incr_con}
Our first observation is that there are arbitrarily large point sets $S$ and triangulations $T\in \mathcal{G}_1(S)$ such that $T$ cannot be augmented into a 4-connected biplane graph on $S$ by adding a plane graph. These special triangulations are the wheels and fans. Recall that a \emph{wheel} is a triangulation on $n$ points such that $n-1$ points are in convex position and one point (the \emph{center} of the wheel) lies in the interior of $\ch(S)$, the points on the convex hull induce a cycle on the boundary of $\ch(S)$ and the interior point is joined to all other $n-1$ points. A \emph{fan} is a triangulation on $n$ points in convex position, where one point (called \emph{center}) is joined to all other $n-1$ points. Every wheel is 3-connected and every fan is 2-connected. We show below that their connectivity cannot be increased to 4 by augmenting them with a set of pairwise noncrossing edges. Moreover, a fan cannot be augmented into a 4-connected biplane graph, even in the case that the added graph is not plane.

\begin{lemma}\label{lem:fan}
Let $T=(S,E)$ be a fan. Then the vertex-connectivity of every biplane graph $G=(S,E\cup E')$ is at most 3.
\end{lemma}
\begin{proof}
Let $T=(S,E)$ be a fan centered at $v_0$, with $\ch(S)=(v_0,\ldots ,v_n)$. Consider
a biplane graph $G=(S,E\cup E')$ and suppose it is 4-connected. We may assume that $G$ is maximal biplane, and by Lemma~\ref{lem:tri}, it is the union of two triangulations $T_1=(S,E_1)$ and $T_2=(S,E_2)$. Note that the two triangulations share all convex hull edges.

First suppose that all edges of the fan are in one of the two triangulations, say $T_1=T$. Then every vertex $v_i$, except for the center of the fan, has degree at most 3 in $T_1$, because the only adjacent edges to $v_i$ are the convex hull edges $(v_i, v_{i-1})$ and $(v_i,v_{i+1})$, and the chord $(v_i, v_0)$ (if it exists). On the other hand, any triangulation on a convex point set has at least two vertices of degree 2, so the only edges adjacent to these vertices of degree 2 are convex hull edges. Hence, the triangulation $T_2$ has at least two vertices of degree 2. At most one of these vertices may be $v_0$, but the other one must be a vertex of degree at most 3 in $T$. Since the convex hull edges are common to both triangulations, we conclude that at least one vertex in $S$ has degree at most 3 in $G$.

Now assume that not all edges of the fan are in $T_1$. Then the fan has consecutive edges that belong to different triangulations, say $v_0v_i\in E_1$ and $v_0v_{i+1}\in E_2$ for $1<i<n-1$. Then $\{v_0,v_i\}$ is a 2-vertex cut in $T_1$, and $\{v_0,v_{i+1}\}$ is a 2-vertex cut in $T_2$. No edge of $G$ can cross both $v_0v_i\in E_1$ and $v_0v_{i+1}\in E_2$. Therefore $\{v_0,v_i,v_{i+1}\}$ is a 3-vertex cut of $G$.
\end{proof}

\begin{lemma}\label{lem:wheel}
Let $T=(S,E)$ be a wheel.
For every triangulation $T'=(S,E')$, the vertex-connectivity
of the biplane graph $G=(S,E\cup E')$ is at most 3.
\end{lemma}
\begin{proof}
Let $pq$ be an edge in $E'\setminus E$. Necessarily, $p$ and $q$ are nonconsecutive vertices on the boundary of $\ch(S)$ because $T$ is a wheel. Since $T'$ is a plane graph, no edge in $E'$ crosses $pq$, and so $\{p,q\}$ is a 2-vertex cut in $T'$. Every path in $G$ between two vertices on opposite sides of the line $pq$ must pass through $p$, $q$, or the center of the wheel. Hence $G$ has a 3-vertex cut.
\end{proof}

\paragraph{Remark.} Note that in some cases, a wheel $T=(S,E)$ can be augmented into a 4- or even 5-connected biplane graph $G=(S,E\cup E')$, but then some of the new edges in $E'$ must cross each other (that is, $T'=(S,E')$ is not a plane graph). For example, consider a wheel $T=(S,E)$ on $2n+1$ points where $\ch(S)$ is a regular $2n$-gon and the center of the wheel is located at the centroid of the polygon. Then $T$ can be augmented into a 5-connected biplane graph $G=(S,E\cup E')$, where $E'$ consists of two cycles, each of $n$ vertices, formed by the odd and even points of $2n$-gon, respectively.

\paragraph{} Recall that every maximal biplane graph on $n\geq 3$ points is 3-connected (Lemma~\ref{lem:flip}), hence Lemmas~\ref{lem:fan} and~\ref{lem:wheel}  are tight. We next show that wheels and fans are the only exceptions, that is, any other triangulation on $n\geq 4$ vertices can be augmented to a 4-connected biplane graph by adding a plane graph. To that end, we use the following auxiliary lemma.

\begin{lemma}\label{lem:point}
Let $S$ be a set of $n\geq 5$ points in general position such that $u$, $v$, and $w$, are consecutive vertices of $\ch(S)$. Let $T$ be a triangulation on $S$ containing the empty triangles $\Delta = (u,v,w)$ and $\Delta_1 = (u,v',w)$. Then, edge $uw$ is flippable in $T$ and, after flipping $uw$, one of the edges $uv'$ or $v'w$ is also flippable in the new triangulation.
\end{lemma}
\begin{proof}
Since $u$, $v$ and $w$ are consecutive vertices of $\ch(S)$, all points in $S$ lie in the closed angular domain $\angle (u,v,w)$. Hence, the points $u$, $v$, $w$ and $v'$ are in convex position, and edge $uw$ is flippable. Replace edge $uw$ with the edge $vv'$ to obtain a new triangulation $T'$.

In $T'$, the triangles $(u,v,v')$ and $(v,w,v')$ are adjacent to edges $uv'$ and $v'w$, respectively. Let $\Delta _2 = (u,x,v')$ and $\Delta _3 = (w,y,v')$ be (if they exist) the other two triangles of $T'$ adjacent to edges $uv'$ and $v'w$, respectively (possibly $x=y$). As $T'$ is a triangulation and $n\ge 5$, at least one of $\Delta_2$ and $\Delta_3$ must exist. Assume first that only one of $\Delta_2$ and $\Delta_3$, say $\Delta_2$, exists. Then $v'$ is a vertex of the convex hull, and so $u$, $v$, $v'$ and $x$ are in convex position. Consequently, edge $uv'$ is flippable. Next, assume that both $\Delta_2$ and $\Delta_3$ exist. Notice that the edges $v'v$, $v'u$, $v'x$, $v'y$, and $v'w$ appear in this circular order around $v'$ (possibly $v'x=v'y$).
Let $\ell$ be the line spanned by $v$ and $v'$. If $x$ and $u$ are in the same halfplane defined by $\ell$, then $(v,v',x,u)$ is a convex quadrilateral and, consequently, edge $uv'$ is flippable. Otherwise, $x$, $y$ and $w$ must be in the same halfplane defined by $\ell$. In this case, $(v,v',w,y)$ is a convex quadrilateral and edge $wv'$ is then flippable.

\end{proof}

We also recall the characterization of 2- and 3-vertex cuts in a triangulation $T=(S,E)$ on $n\geq 5$ vertices. A set $\{u,v\}\subset S$ is a 2-vertex cut in $T$ if and only if $uv$ is a chord of $\ch(S)$. A \emph{separating triangle} of $T$ is a triangle containing a vertex in its interior and another vertex in its exterior. We define a \emph{bichord} of $T$ as a path of length two between two vertices of $\ch(S)$ such that no edge of the path lies on the boundary of $\ch(S)$ and there is a vertex on each side of the path. If the middle vertex of the bichord is an interior vertex, then the bichord splits $S$ into two nonempty sets. If the middle vertex is a vertex on $\ch(S)$, then the bichord is the union of two chords and splits $S$ into three nonempty sets.
A set $\{u,v,w\}\subset S$ is a 3-vertex cut in $T$ if and only if they form a separating triangle, a bichord or some of them form a chord.

In all the cases (2- and 3-vertex cuts), $S$ is split into several nonempty subsets by a chord, a separating triangle or a bichord, such that no edge of $T$ connects vertices in distinct subsets. When we augment a triangulation $T$ into a 4-connected biplane graph $G$, we must add edges connecting vertices belonging to these subsets. For example, a separating triangle consisting of three chords splits $S$ into four nonempty subsets. In this case, if we remove the triangle, then at least three new edges crossing the chords must be added to reconnect the graph. Observe that the edges of this triangle also form three bichords.

We say that an edge \emph{properly} crosses a bichord or a separating triangle when the edge intersects it in exactly one internal point (a chord is always crossed properly). We wish to augment a triangulation $T=(S,E)$ on $n\geq 5$ vertices into a 4-connected biplane graph $G=(S,E\cup E')$ by adding a plane graph $T'=(S,E')$ (not necessarily a triangulation). Since $T$ is 2-connected, any 2- or 3-vertex cut of an augmentation $G=(S,E\cup E')$ is a 2- or 3-vertex cut of $T$. Therefore, it is not difficult to check that $G$ is 4-connected if and only if (i) every separating triangle and every bichord of $T$ are properly crossed by at least one new edge, (ii) every chord of $T$ is properly crossed by at least two new edges, and (iii) when $S$ contains at least two points on each side of a chord of $T$, then the chord is crossed by at least two nonadjacent new edges (otherwise the chord and a common endpoint of crossing edges would be a 3-vertex cut).

We are now ready to augment any triangulation, other than the wheel or the fan,
by a plane graph to a 4-connected biplane graph.

\begin{theorem}\label{theo:augment}
Let $S$ be a set of $n\ge 6$ points in convex position or a set of $n\ge 5$ points not in convex position. For every triangulation $T=(S,E)$ other than the wheel and the fan, there is a plane graph $T'=(S,E')$ such that $G=(S,E\cup E')$ is 4-connected.
\end{theorem}
\begin{proof}
We consider two cases depending on the vertex connectivity of $T$.

\paragraph{Case~1: $T$ is 3-connected.} In this case, since there are no chords, it is enough augment $T$ such that every separating triangle and every bichord of $T$ is properly crossed by at least one new edge.
If there is a vertex $v$ that does not belong to any 3-vertex cut, then augment $T$
with a star $T'=(S,E')$ centered at $v$. Then in $(S,E\cup E')$, every separating
triangle and bichord is properly crossed by an edge incident to $v$.

Assume now that every vertex in $S$ is part of a 3-vertex cut of $T$.
In this case, we start by showing that $T$ has no separating triangle. Suppose, to the contrary, that there is a separating triangle in $T$. Let $\Delta$ be a separating triangle that contains the minimum number of vertices in its interior, and let $v$ be a vertex in the interior of $\Delta$. Then $v$ cannot be part of any separating triangle by the minimality of $\Delta$. If $v$ is part of a bichord, then the two endpoints of the bichord must be vertices of $\Delta$, and so an edge of $\Delta$ is a chord of $\ch(S)$, contradicting our assumption that $T$ is 3-connected. We conclude that $T$ has no separating triangle.

\begin{figure}[!htb]
\centering
\includegraphics[scale=0.6]{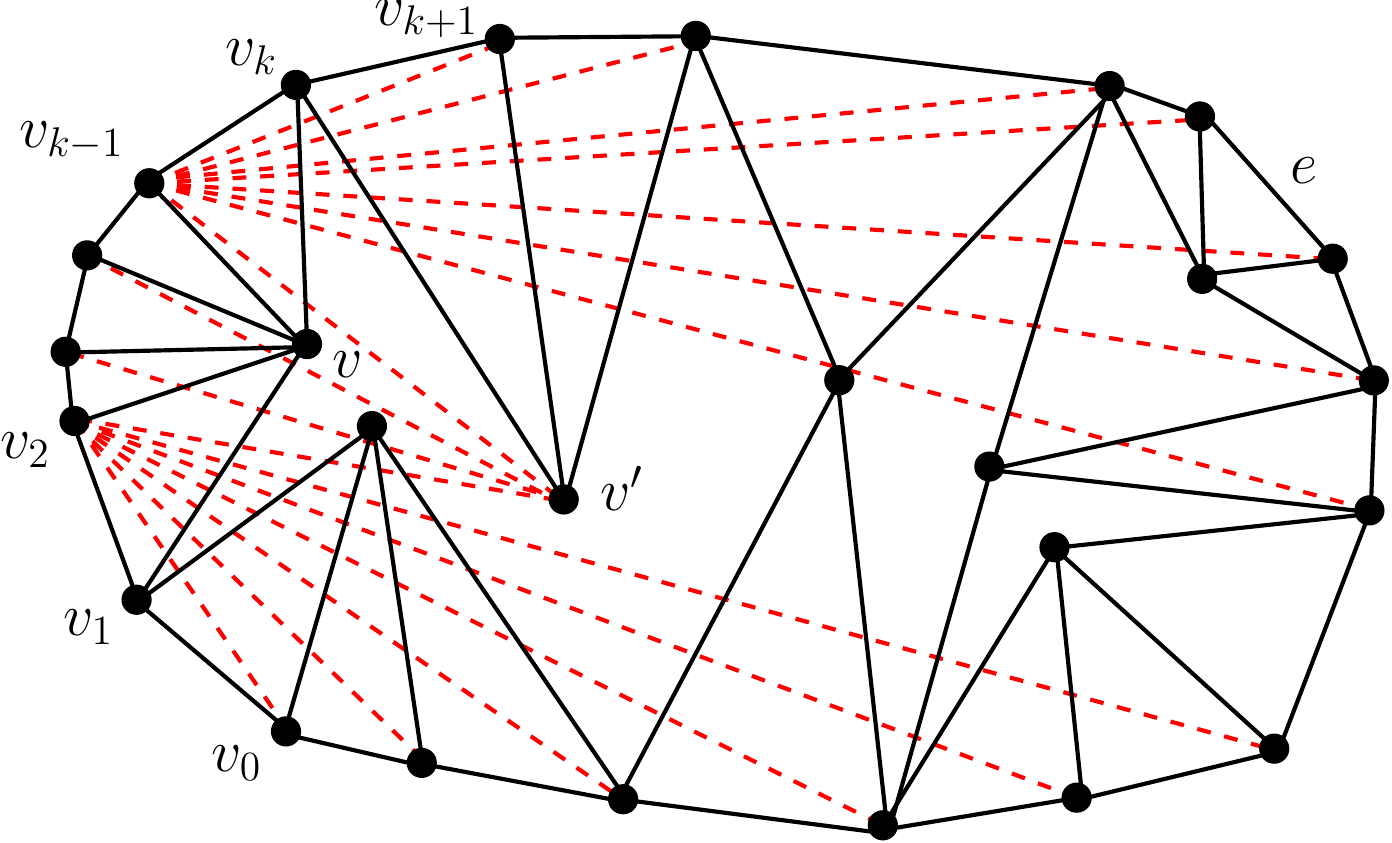}
\caption{A spanning subgraph $H$ formed by the union of bichords. All new edges (dashed edges) are incident to $v_2$, $v_{k-1}$ or $v'$.}\label{pic:general3to4}
\end{figure}

We are left with the case in which every 3-vertex cut is a bichord and every point belongs to some bichord. Denote by $H$ the subgraph of $T$ formed by the union of all bichords. The graph $H$ is the union of stars centered at interior points (each star is the union of all bichords with a common middle vertex). Refer to Figure~\ref{pic:general3to4}. Each star divides $\ch(S)$ into several \emph{sectors}. There are at least two interior points, since $T$ is not a wheel. Let $e$ be an arbitrary edge of $\ch(S)$, and let $v$ be an interior vertex such that the sector of $v$ containing $e$ has maximal area. The maximality implies that no other sector of $v$ contains any bichord.  Denote by $v_1, \ldots , v_k$ the neighbors of $v$ in $H$ in clockwise order such that all other interior vertices of $H$ lie in the sector bounded by $vv_1$ and $vv_k$.
Note that $k\geq 4$, since if $k\in \{2,3\}$, then $v_2$ cannot belong to a bichord, contradicting the assumption that every vertex of $G$ is part of some bichord.

We can now define a plane graph $T'$ as follows. Connect $v_2, \ldots , v_{k-1}$ to an arbitrary interior vertex $v'$, $v'\neq v$. Partition the sector bounded by $v'v_2$ and $v'v_{k-1}$ that contains $v_k$ into two convex parts along the angle bisector of $\angle(v_2,v',v_{k-1})$. Connect $v_2$ and $v_{k-1}$, respectively, to all vertices of $\ch(S)$ that lie in the same convex part of the sector (see Figure~\ref{pic:general3to4}). If the angle bisector hits a vertex $w$ of $\ch(S)$ then $w$ can be connected to either $v_2$ or $v_{k-1}$.
The graph $T'$ contains an edge properly crossing every bichord. Therefore $G=(S,E\cup E')$ is 4-connected, as required.

\begin{figure}[!htb]
\centering
\includegraphics[width=0.9\textwidth]{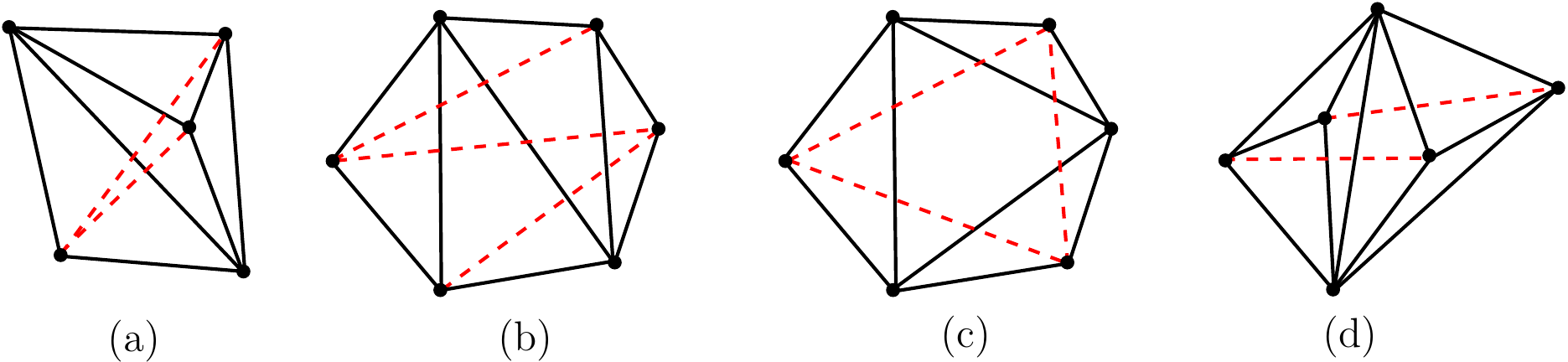}
\caption{(a--c) Base cases when $T$ in not 3-connected. Dashed edges indicate edges in $E'$. (d) A triangulation of 6 points with two interior points on opposite sides of a chord.
}\label{pic:basetheorem3}
\end{figure}

\paragraph{Case~2: $T$ is not 3-connected.}
Since every triangulation on $n\geq 3$ vertices is 2-connected, then $T$ must have a 2-vertex cut. Each 2-vertex cut $\{u,v\}$ corresponds to a chord $uv\in E$ of $\ch(S)$. The chords of $\ch(S)$ decompose $\ch(S)$ into convex regions, called \emph{cells}. We define a dual graph $H$ on the cells as follows: the nodes of $H$ correspond to the cells, and two nodes are joined by an edge in $H$ if and only if the corresponding cells share a chord. Clearly, $H$ is a tree. The cells corresponding to the leaves of $H$ are called \emph{leaf cells}. For a leaf $\ell$ of $H$, denote by $S_\ell\subset S$ the set of vertices that lie in the interior or on the boundary of the leaf cell $\ell$ (including the endpoints of the chord defining the leaf). The \emph{size} of a leaf cell $\ell$ is the cardinality of $S_\ell$. Note that the subgraph $T_\ell$ of $T$ induced by the vertices in $S_\ell$ for a leaf cell $\ell$, is either a triangle (if $|S_\ell|=3$) or a 3-connected triangulation (if $|S_\ell|\geq 4$).
We proceed by induction on $n$, the total number of vertices.

\paragraph{Base cases.}
Suppose $n=5$ and the points are not in convex position. Since $T$ has a chord, $\ch(S)$ has 4 vertices, and there  is one interior point (Figure~\ref{pic:basetheorem3} (a)). Since $\ch(S)$ is a quadrilateral, there is exactly one chord, which determines two leaf cells. The vertices in the two cells induce $K_3$ and $K_4$, respectively,
with 1 and 2 vertices disjoint from the chord. Let $E'$ contain the two edges between the vertices on opposite sides of the chord. Then $G=(S,E\cup E')$ is isomorphic to $K_5$, which is 4-connected.

Suppose $n=6$ and the points are in convex position. Every triangulation $T$, other than a fan, is composed of three chords of $\ch(S)$ forming a triangle or a path. In both cases, there is a set $E'$ of 3 noncrossing edges such that $(S,E\cup E')$ is 4-connected, as indicated in Figures~\ref{pic:basetheorem3}(b) and~6(c).

\paragraph{Induction step.} To prove the induction step we will distinguish two cases, based on whether a leaf cell of size 3 exists or not.

\paragraph{Case~2.1: There is a leaf cell $L$ of size 3.}
First, assume $n\ge 6$ ($n\ge 7$ for points in convex position) and suppose that a leaf cell of size 3 exists, i.e. there is a triangle $\Delta = (u,v,w)$ in $T$, with $uw$ being a chord and $v$ being the only vertex (on the boundary of $\ch(S)$) to the left of $\overrightarrow{uw}$. By removing $v$ from $T$, we obtain a triangulation $T_1$ on the $n-1$ remaining points. Three subcases arise depending on whether $T_1$ is a wheel, a fan or neither.

If $T_1$ is a wheel, then let $T'=(S,E')$ be a star centered at $v$. The biplane graph obtained $(S,E\cup E')$ is 4-connected because all bichords of $T$ (which are the bichords of $T_1)$ are properly crossed by $T'$, and the only chord $uw$ is crossed by at least two edges of $T'$.

If $T_1$ is a fan, then, instead of removing $v$, we can remove a different vertex of degree 2 from $T$, obtaining a new 2-connected triangulation other than a fan (otherwise $T$ would be a fan).

\begin{figure}[!htb]
\centering
\includegraphics[width=0.85\textwidth]{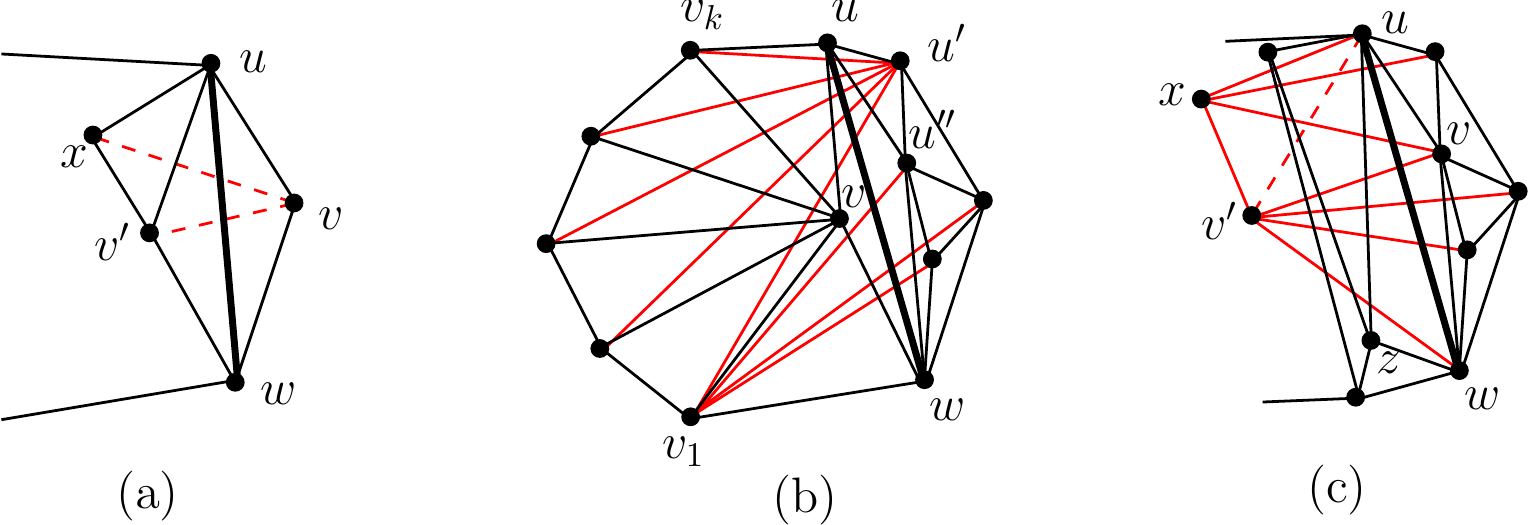}
\caption{Augmenting 2-connected triangulations to 4-connected biplane graphs by crossing chords.
Chord $uw$ is shown as a bold segment.}\label{pic:theorem3}
\end{figure}

Otherwise, $T_1$ is a 3-connected triangulation other than a wheel or a 2-connected triangulation other than a fan. In this case, by induction, there is a plane graph $T_2$ on the $n-1$ remaining points such that $T_1\cup T_2$ is biplane and 4-connected. We may assume that $T_2$ is a triangulation.

Now, we construct a triangulation from $T_2$ by adding vertex $v$ and some incident edges, and possibly deleting some of the edges of $T_2$, as explained next. Refer to Figure~\ref{pic:theorem3}(a). Let $\Delta_1 = (u, w, v')$ be the triangle of $T_2$ adjacent to $uw$ and let $\Delta_2 = (u,x,v')$ and $\Delta_3 = (w,y,v')$ be (if they exist) the two triangles of $T_2$ adjacent to triangle $\Delta_1$.  Note that at least one of them must exist. First, we add $\Delta$ to $T_2$, obtaining a triangulation $T'_2$ on the $n$ points. By Lemma \ref{lem:point}, we can flip edge $uw$ in $T'_2$ and, after flipping it, we can flip one of the edges $uv'$ and $v'w$.
Assume without loss of generality that we flip $uv'$.
Then, from $T'_2$, we remove edges $uw$ and $uv'$ and we add edges $vv'$ and $vx$, obtaining a new triangulation $T'$ on the $n$ points.
By Property~\ref{pro:property2+}, $T\cup T'$ is 4-connected. Indeed, we have obtained $T\cup T'$ from  the 4-connected graph $T_1\cup T_2$ by removing at most one edge, $uv'$ (since edge $uw$ is always in both $T_1$ and $T_2$) and adding a new vertex $v$ joined to $u$, $v'$ and two additional vertices, $x$ and $w$. This completes the proof of Case~2.1.

\paragraph{Case~2.2: There is no leaf cell of size 3.}
Our proof in this case is similar to Case~2.1, but the constructions are now a bit more complicated.
Since $E$ contains at least one chord of $\ch(S)$, the convex hull has at least 4 vertices and there are at least two leaf cells.
Since no leaf cell is a triangle, there are at least two points in the interior of $\ch(S)$.

Assume first that $n=6$. In this case, there are exactly four hull vertices, hence there is a unique chord, and exactly one interior vertex on each side of the chord (Figure~\ref{pic:basetheorem3}(d)). Consider the four vertices disjoint from the chord, on each side: they admit two disjoint edges between vertices on opposite sides of the chord. These edges augment $T$ into a 4-connected biplane graph, as required.

Assume now $n\ge 7$. Let $\ell$ be a leaf cell of minimal size (i.e., with minimum number of vertices). This leaf cell is defined by a chord $uw$, and its size $|S_\ell|$ is at least 4. If we remove the points in $S_\ell$ except for $u$ and $w$, then all vertices of the second smallest leaf cell survive, and so we are left with at least $|S_\ell| \geq  4$ vertices. We obtain a triangulation $T_1$ on the $n-|S_\ell|+2$ remaining points. Note that, if $|S_\ell| = 4$, then $n-|S_\ell|+2 \ge 5$ because $n\ge 7$, and if $|S_\ell| \ge 5$, then $n-|S_\ell|+2 \ge |S_\ell| \ge 5$. Therefore, if $T_1$ is neither a wheel nor a fan, we can apply induction to the triangulation $T_1$.

First, observe that $T_1$ cannot be a fan, otherwise it would have two leaves of size 3, and one of them would be a leaf cell of size 3 in $T$, contradicting the assumption that there is no such leaf cell. Assume that $T_1$ is a wheel. Refer to Figure~\ref{pic:theorem3}(b). Let $v$ denote the center of the wheel $T_1$, and let the neighbors of $v$ be denoted $w, v_1, \ldots, v_k, u$ in clockwise order. Note that $k \geq 2$ because $n-|S_\ell|+2 \ge 5$ and so the degree of $v$ is at least 4. Let $h$ be the line through $v_1$ and $u$, and let $u'\ne w$ be the first point in $S_\ell$ hit when we rotate $h$ clockwise about $v_1$. Now, let $T'$ contain edges connecting $u'$ to $v_1, \ldots , v_k$; and $v_1$ to all the vertices in $S_\ell\setminus \{u, w\}$. Since $T_\ell$ (the graph induced by $S_\ell$ in $T$) is 3-connected, then by connecting each vertex in $S_\ell$ (except for $u$ and $w$) to a vertex outside of $S_\ell$, each separating triangle and bichord of $T_\ell$ is properly crossed. In addition, every bichord incident to $v$ is properly crossed by an edge of type $u'v_j$. Finally, the only chord of $T$, the chord $uw$, is properly crossed by two new edges, for example, edges $v_1u''$ and $v_ku'$ (where $u''$ is an arbitrary vertex of $\ell$ other than $u$, $w$ and $u'$, which must exist because $|S_\ell|\ge 4$). Therefore, $(S,E\cup E')$ is 4-connected.

We can now assume that $T_1$ is a 3-connected triangulation other than a wheel or a 2-connected triangulation other than a fan. By induction, there is a plane graph $T_2$ on the $n-|S_\ell|+2$ remaining points such that $T_1\cup T_2$ is a 4-connected biplane graph. We may assume that $T_2$ is a triangulation. We modify $T_2$ to construct a new plane graph $T'$ on all $n$ points as follows (see Figure \ref{pic:theorem3}(c)). Similarly to Case~2.1, let $\Delta_1 = (u, w, v')$ be the triangle of $T_2$ adjacent to the chord $uw$ and let $\Delta_2 = (u,x,v')$ and $\Delta_3 = (w,y,v')$ be (if they exist) the two triangles of $T_2$ adjacent to triangle $\Delta_1$. Note that at least one of them must exist. Let $\Delta = (u,w,v)$ be the triangle of $T_\ell$ adjacent to edge $uw$. Vertex $v$ must be an interior vertex because $T_\ell$ is 3-connected. By adding $\Delta$ to $T_2$, we obtain a new triangulation $T'_2$ and, again, by Lemma \ref{lem:point}, we can flip edge $uw$ in $T'_2$. We can then flip one of the edges $uv'$ and $v'w$. Assume without loss of generality that we flip $uv'$. We construct a plane graph $T'$ from $T_2$ on the $n$ points by removing edges $uw$ and $uv'$, adding edges $xv$ and $v'v$, and connecting all remaining points in $S_\ell$ to one of $x$ and $v'$, depending on which side of the angle bisector of $\angle(x,v,v')$ they lie on. We have added an edge of type $v'v''$ or $xv''$ adjacent to each vertex $v''$ in $S_\ell$, where $v''\not\in \{u,w\}$.

We claim that $T\cup T'$ is 4-connected. We need to show that every chord, every bichord and every separating triangle of $T$ is properly crossed. By induction, $T_1\cup T_2$ is 4-connected. In the case that edge $uv'$ is removed from the 4-connected graph $T_1 \cup T_2$, the path $(u,v,v')$ establishes a new connection from $u$ to $v'$. Note that, since triangle $\Delta_1$ is empty, any chord properly crossed by edge $uv'$ is also properly crossed by  edge $v'v$. Consequently, every separating triangle, every bichord and every chord of $T_1$, is properly crossed at least once or twice, as required. However, it is possible that a bichord of $T$ consists of two edges of $T_1$ but it is not a bichord in $T_1$. One possibility for such a bichord is $(u,z,w)$, where $\Delta_0=(u,w,z)$ is the triangle of $T_1$ adjacent to the chord $uw$ and $z$ is an interior point (possibly, $z=v'$). See Figure~\ref{pic:theorem3}(c). This bichord (if it exists) is properly crossed by the edge $xv$. The other possibility is that edge $uw$ belongs to a bichord of $T$ consisting of two chords, say $wu$ and $uw'$. In this case, since there are at least two points on each side of $uw'$ (there are no leaf cells of size 3 on one of the sides, and $w$ and $z$ are on the other side), by induction, there are at least two disjoint edges crossing $uw'$, one of them properly crossing the bichord. Moreover, every separating triangle and every bichord in $T_\ell$ is properly crossed by an edge of type $v'v''$ or $xv''$, where $v''$ is in $S_\ell$. Finally, we show that the chord $uw$ is properly crossed by two disjoint edges. For an arbitrary vertex $v''$ in $S_\ell\setminus \{u,v,w\}$, if $v''$ is connected to $x$, then both $xv''$ and $v'v$ cross $uw$, and if $v''$ is connected to $v'$, then both $xv$ and $v'v''$ cross $uw$.
\end{proof}

We have seen (Lemmas~\ref{lem:fan} and~\ref{lem:wheel}) that a fan and a wheel, respectively, are 2- and 3-connected triangulations that cannot be augmented to a 4-connected biplane graph by adding a second triangulation. We now show that there are 4-connected triangulations that cannot be augmented to 5-connected biplane graphs by adding a second triangulation.

\begin{theorem}\label{lem_no5con}
There exist arbitrarily large point sets $S$ and 4-connected triangulations $T=(S,E)$ such that
for every triangulation $T'=(S,E')$, the biplane graph $(S,E\cup E')$ is not 5-connected.
\end{theorem}

\begin{proof}
Our construction is shown in Figure~\ref{pic:NonAugTo5Con2}, where the initial triangulation $T$ appears in Figure~\ref{pic:NonAugTo5Con2}(a) drawn with black edges. The point set has two main clusters. The top cluster consists of $4k-3$ points in convex position, with $2k$ points in a lower chain $x_1,\ldots ,x_{2k}$, and $2k-3$ additional in an upper chain $x_1,y_1,\ldots, y_{2k-3},x_{2k}$. The parameter $k\in \mathbb{N}$ can be made arbitrarily large. The lower cluster consists of 7 points as shown in Figure~\ref{pic:NonAugTo5Con2}(a). The bottom cluster is sufficiently far below the top cluster so that any new edge between a point $y_i$ and a point in the bottom cluster crosses the edge $x_{i+1}x_{i+2}$. Intuitively, the bottom cluster is a ``big dot'' far below the top cluster. It is not difficult to verify that $T$ is 4-connected (i.e., it has no chords, bichords, or separating triangles).

\begin{figure}[htbp]
\centering
\includegraphics[width=0.96\textwidth]{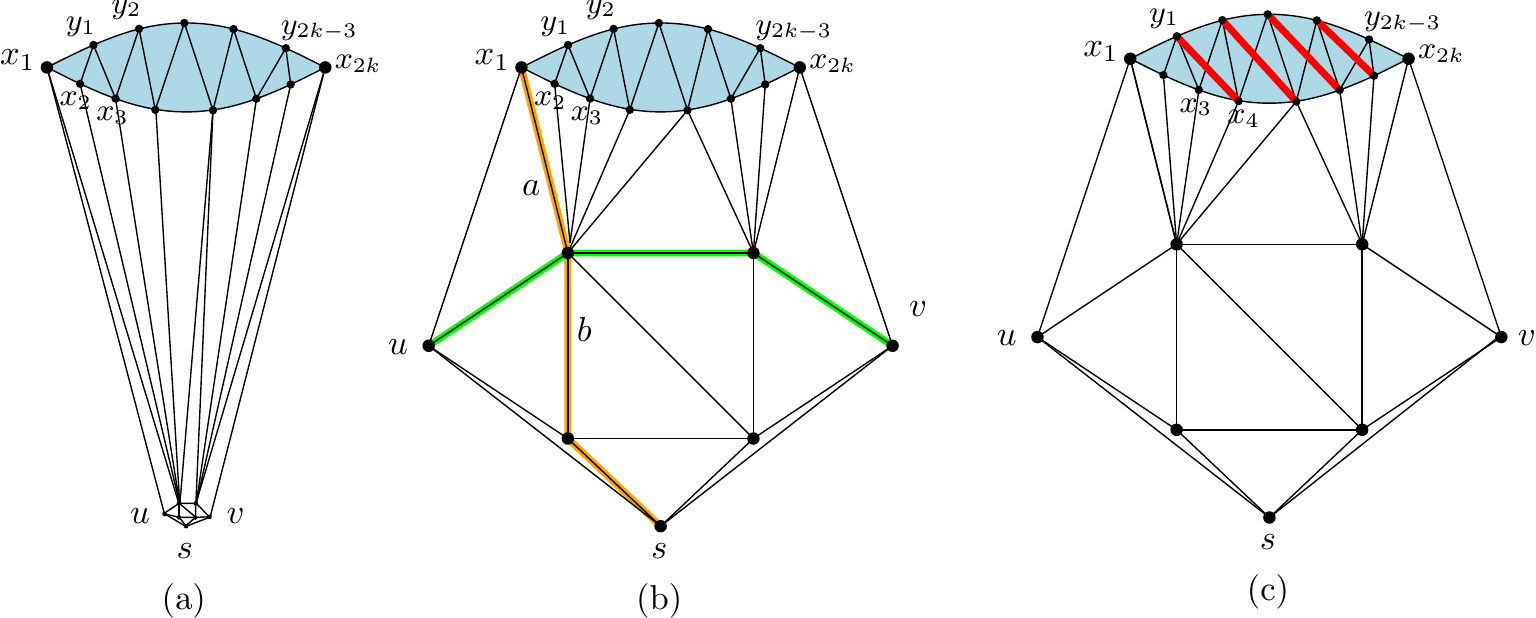}
\caption{
(a) Overview of the construction. The top convex cluster is sufficiently far from the bottom part.
(b) Schematic view of the construction, highlighting two of the four critical trichords of $T$.
(c) To increase the degree of the vertices in the convex cluster, more edges are needed. However, vertex $y_{2k-3}$ cannot connect to any other vertex.
}\label{pic:NonAugTo5Con2}
\end{figure}

Our argument crucially depends on four 4-vertex cuts in $T$, each consists of the vertices of a path with 3 edges. We call these the four \emph{critical trichords} of $T$. Two such critical trichords are highlighted in Figure~\ref{pic:NonAugTo5Con2}(b): one from $u$ to $v$, and one from $s$ to $x_1$. The other two are symmetric (one also goes from $u$ to $v$, and the other from $s$ to $x_{2k}$).

Suppose, to the contrary, that there is a triangulation $T'=(S,E')$ such that $G=(S,E\cup E')$ is 5-connected. Then the vertices on opposite sides of each critical trichord must be connected by at least one edge in $E'$. If an edge $e\in E'$ crosses the trichord from $s$ to $x_1$, then $e$ must be incident to $u$, and it must cross either edge $a$ or edge $b$. If $e$ crosses $b$, it is easy to verify that, no matter what the other endpoint of $e$ is (there are only three possibilities), it is impossible to add pairwise noncrossing edges that cross both trichords between $u$ and $v$, and the trichord from $s$ to $x_{2k}$, but do not cross $e$. We conclude that $e$ must cross $a$ and hence is incident to some vertex in the top cluster.

If $G$ is 5-connected, every vertex must have degree at least 5 in $G$, and so every vertex in $\{x_1, y_1, y_2, \ldots, y_{2k-3}, x_{2k}\}$ must be incident to at least one new edge in $E'\setminus E$. Consider an edge $e_1\in E'\setminus E$ incident to $y_1$. If $e_1$ connects $y_1$ to any vertex in $\{y_3,\ldots,y_{2k-3}\}$, then the vertices on or above $e_1$ form a convex polygon within in the top cluster whose triangulation in $T'$ would necessarily contain two ears; therefore at least one vertex $y_i$ would not be incident to any new edge. If $e_1$ connects $y_1$ to any vertex in the bottom cluster, then $e_1$ crosses edge $x_2x_3$. Since $e$ and $e_1$ do not cross, we have either $e=ux_2$ or $e=e_1=uy_1$: In any case, $x_1$ cannot be incident to any new edge in $E'\setminus E$.
Therefore, $e_1$ must connect $y_1$ to some vertex in the chain $(x_4,x_5,\ldots , x_{2k})$.

Assume that $e_1$ connects $y_1$ to $x_{j(1)}$ for some $j(1)\geq 4$. Using a similar reasoning, $y_2$ should connect to $x_{j(2)}$ for $j(2)\geq 5$, and in general, $y_i$ should connect to $x_{j(i)}$ for some $j(i)\geq i+3$, see Figure~\ref{pic:NonAugTo5Con2}(c). Since the top chain has three fewer vertices than the bottom chain, then $y_{2k-3}$ cannot be connected to any other vertex by a new edge, and its vertex degree is 4 in $G$, contradicting our assumption that $G$ is 5-connected.

\end{proof}

\subsection{Minimal Augmentation}
\label{sec:aug_minimal}
Given a plane graph $G=(S,E)$, we wish to augment $G$ with a minimal set of new edges $E'$ such that we obtain a $k$-connected biplane graph $G=(S,E\cup E')$ for some target value $k$.
In this section we present an efficient solution (Lemma~\ref{lem_minaugment})
when $k=3$ and $G$ is a triangulation. We start with a helpful lemma about
augmenting a plane tree to 2-edge-connectivity.

\begin{lemma}\label{lem:tree}
Given a plane tree $H=(S,E)$ with $n$ vertices and $m$ leaves, let $L\subset S$ be the set of $m$ leaves of $H$.
In  $O(n+m\log m)$ time, one can find a set $E'$ of $\lceil m/2\rceil$ pairwise noncrossing edges
among the leaves of $H$ such that $(S,E\cup E')$ is a 2-edge-connected
biplane graph. Moreover, if the $m$ leaves are in convex position,
then $E'$ can be found in $O(n)$ time, provided that the clockwise ordering of the leaves along their convex hull is given.
\end{lemma}

It is well known (see for example \cite{eswaran}) that an abstract tree with $m$ leaves can be augmented to a 2-edge-connected graph by adding $\lceil m/2\rceil$ new edges among its leaves. However, establishing the noncrossing condition requires a proof.

\begin{proof}
Choose a root $r\in S$ arbitrarily, and let $T$ be the rooted tree obtained from $H$. We denote by $a\prec b$ if $a$ is a descendent of $b$ in the rooted tree $T$. For two leaves $u,v\in L$, denote by $\lca(u,v)$ their \emph{lowest common ancestor} in $T$. For a rooted tree with $O(n)$ vertices, there is a data structure that can report $\lca(u,v)$ for any query vertex pair $\{u,v\}$ in $O(1)$ time after $O(n)$ time preprocessing~\cite{BF00,HT84,SV88}.

We construct the set of new edges recursively. Initialize $E'$ to be the empty set. If $|L|\in \{2,3\}$, then add an arbitrary spanning tree of $L$ with $\lceil |L|/2\rceil$ edges to $E'$. While $|L|>3$, repeat the following loop: Let $v$ be an arbitrary vertex of $\ch(L)$ other than the root, and let $u\in L$ and $w\in L$ be the two vertices of $\ch(L)$  adjacent to $v$. If $\lca(u,v)\succeq \lca(v,w)$, then put edge $uv$ into $E'$, and remove both $u$ and $v$ from $L$. Otherwise, when $\lca(u,v)\prec\lca(v,w)$, put edge $vw$ into $E'$, and remove both $v$ and $w$ from $L$.

In each loop, the algorithm selects an edge from the boundary of the convex hull of $L$, which cannot cross any edge selected in a later loop. It follows that $E'$ consists of pairwise noncrossing edges. The algorithm adds one edge for each pair of vertices until $|L|$ drops below 4, and then it uses $\lceil |L|/2\rceil$ edges. So we have $|E'|=\lceil m/2\rceil$.

It remains to show that $(S,E\cup E')$ is 2-edge-connected. Let $ab$ be an edge of $T$ with $a\prec b$, and let $L_{ab}$ denote the set of leaves that are the descendants of $b$. Consider the loop of the algorithm in which the last leaf of $L_{ab}$ is removed from $L$. The algorithm connects a leaf in $L_{a,b}$ to another leaf, which cannot be in $L_{a,b}$
(recall that the algorithm compares two alternatives). Thus in this loop, the algorithm adds an edge
that induces a cycle containing $ab$. Therefore, every edge $ab$ is contained in a cycle in $(S,E\cup E')$, thus it is 2-edge-connected. The convex hull $\ch(L)$ can be maintained by the semi-dynamic data structure by Hershberger and Suri~\cite{HS92} in $O(m\log m)$ time, hence the total running time is  $O(n+m\log m)$.

Finally, observe that, when the $m$ leaves are in convex position, $\ch(L)$ can be updated in constant time after removing two consecutive vertices of $\ch(L)$, without using the semi-dynamic data structure by Hershberger and Suri. Therefore, in this case, the set $E'$ of pairwise noncrossing edges can be found in $O(n)$ time, after computing $\ch(L)$ in $O(m\log m)$ preprocessing time.
 \end{proof}

\begin{lemma}\label{lem_minaugment}
Given a triangulation $G=(S,E)$, with $n\geq 3$ vertices,
one can find a minimal set of edges $E'$ such that
$G'=(S,E\cup E')$ is a 3-connected biplane graph in $O(n)$ time, after computing $\ch(S)$ in $O(n\log n)$ preprocessing time.
\label{lem:min3connected}
\end{lemma}
\begin{proof}
If the given triangulation $G$ is not 3-connected, then it must contain 2-vertex cuts. Recall that a set $\{u,v\}\subset S$ is a 2-vertex cut if and only if $uv\in E$ is a chord of $\ch(S)$. The biplane graph $G'=(S,E\cup E')$ will be 3-connected if each chord of $\ch(S)$ in $E$ is crossed by at least one edge in $E'$.

Similar to the proof of Theorem~\ref{theo:augment}, we construct a dual graph $H$. The chords of $\ch(S)$ in $E$ decompose $\ch(S)$ into convex cells. The nodes of $H$ correspond to the cells, and two nodes are joined by an edge in $H$ if and only if the corresponding cells share a chord. Clearly, $H$ is a tree (see the thick edges in Figure~\ref{pic:matching}), and can be easily constructed in $O(n)$ time from $G$.
The leaves of $H$ correspond to leaf cells. Each leaf $\ell$ of $H$ is associated with the set of vertices $R_{\ell}\subset S$ that lie in the interior or on the boundary of the cell, excluding the endpoints of the chord on the boundary of the cell. Note that distinct leaves of $H$ are associated with disjoint vertex sets (i.e., $R_{\ell}\cap R_{\ell'}=\emptyset$ for $\ell\neq \ell'$).
Consider the chords that lie on the boundaries of the leaf cells. These chords are in convex position, thus any new edge can cross at most two of them. It follows that we need to add at least $\lceil m/2 \rceil$ new edges, where $m$ denotes the number of leaves of $H$.

We now show that $\lceil m/2\rceil$ new edges suffice, and can be computed in linear time.
For each leaf $\ell$ of $H$, pick a point $v_{\ell}\in R_{\ell}$ on the boundary of $\ch(S)$.
We refer to this point as the \emph{representative} of $\ell$.
Embed $H$ in the plane such that every
leaf $\ell$ is embedded at point $v_\ell$, and every nonleaf node is embedded at an arbitrary
point in the interior of its cell. Clearly, this embedding can be constructed in $O(n)$ time, after computing $\ch(S)$.

\begin{figure}[tb]
\centering
\includegraphics[width=0.5\textwidth]{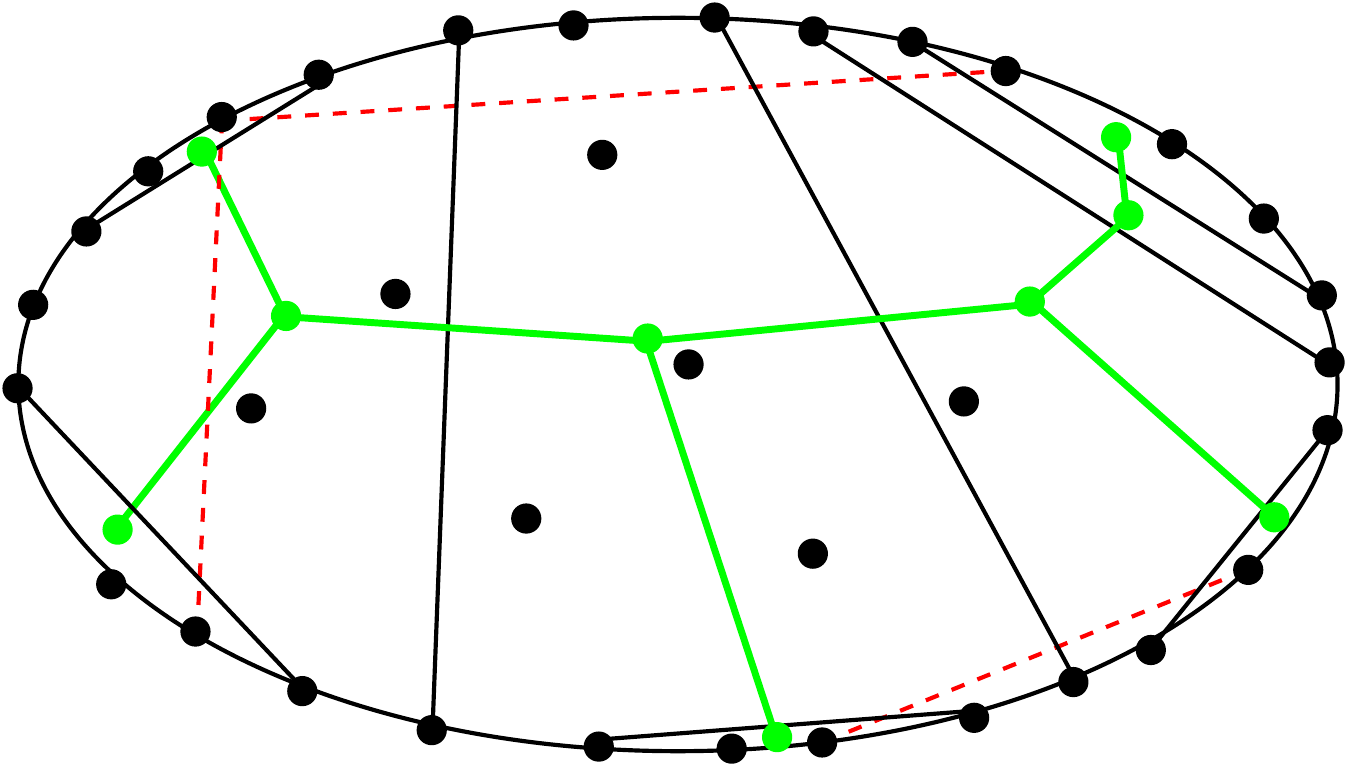}
\caption{A point set, the associated graph $H$ (thick edges), and the additional edges to obtain 3-connectivity (dashed edges).}
\label{pic:matching}
\end{figure}

By Lemma~\ref{lem:tree}, the embedding of $H$ can be augmented to a 2-edge-connected graph $H'$ by a set $E'$
of $\lceil m/2\rceil$ noncrossing edges (dashed edges in Figure~\ref{pic:matching}) between the leaves (that are representatives
from $S$) in $O(n)$ time, because all the representatives are in convex position. We claim that $(S,E\cup E')$ is 3-connected. Every 2-vertex cut of $G$ is a chord $c$ of $\ch(S)$ that corresponds to an edge $e_c$ of the embedding of $H$. When the embedding of $H$ is augmented to $H'$, $e_c$ becomes part of at least one cycle in $H'$, and each of these cycles contains exactly one edge from $E'$. Observe that, if $C$ is a cycle in $H'$ that contains both  $e_c$ and a new edge $e\in E'$, then all edges of $C$ (except for $e$) correspond to chords, which are crossed by edge $e$. In particular, $e$ crosses chord $c$. Since $H'$ is 2-edge-connected, every edge of the embedding of $H$ belongs to a cycle in $H'$, so every chord of $\ch(S)$ is crossed by some edge of $E'$.
\end{proof}

In the above argument, we have added one edge for every two leaf cells of $H$. Associate to each leaf $\ell$ of $H$ all vertices in $R_\ell$ and the two endpoints of the chord. Every vertex of $\ch(S)$ is the endpoint of at most two chords that bound leaf cells. Therefore, one can assign to each leaf cell at least two vertices (a vertex of $R_\ell$ and half of each endpoint of the chord on the boundary of the leaf). It follows that there are at most $\lfloor n/2\rfloor$ leaf cells. This bound is tight, since the chords of the leaf cells may form a cycle of $n/2$ edges.
We obtain an upper bound for the total number of edges added in Lemma~\ref{lem_minaugment}.

\begin{corollary}
Every triangulation $G=(S,E)$ on $n\geq 3$ points can be augmented to a 3-connected biplane
graph by adding at most $\lceil \lfloor n/2\rfloor /2\rceil= \lfloor \frac{n+2}{4} \rfloor$
new edges.
\end{corollary}


\section{Conclusions}\label{sec:conclusion}

We have presented several results on the maximum vertex-connectivity attained
by biplane graphs on a point set $S$, either by constructing a graph from scratch
(starting from the empty graph) or by combining a given plane graph with a new
plane graph. Our proofs are constructive and lead to polynomial-time algorithms for
constructing the new edges. Moreover, when the starting graph is a triangulation, we have also presented an efficient algorithm
for finding the \emph{minimum} number of edges needed to augment the triangulation to a 3-connected biplane graph.

Our Theorem~\ref{theo:5conngeneric} shows that $\G$ contains a 5-connected graph
if the point set $S$ contains 14 points in convex position, which is guaranteed
only for very large sets (roughly $1.3\cdot 10^6$ points or more). Recent
research indicates that the threshold can be reduced to 137 and perhaps to 27
using a so-called \emph{U-condition}~\cite{GHTV13} or combining 4-connected plane
graphs~\cite{GHTV} with stars when no 14 points are in convex position. The improved
bound heavily relies on our Theorem~\ref{theo:5conngeneric}, and will be
the subject of a future paper.

In Theorem~\ref{theo:augment}, we have shown that every triangulation (other than the wheel and the fan)
can be augmented to a 4-connected biplane graph by adding a second plane graph on the same point set;
but a second layer is not always sufficient to augment a plane graph to a 5-connected biplane graph.
In general, we do not know whether every sufficiently large plane triangulation, other than the fan,
can be augmented to a 5-connected biplane graph when the new edges are not required to form a plane graph.

Several computational problems related to our results remain open. Is there a polynomial-time algorithm
that, given a point set $S$, finds a 5-connected biplane graph with the \emph{minimum} number of edges
or reports that none exists? Is there a polynomial-time algorithm for finding the \emph{minimum} number of
edges to augment a given 3-connected plane graph with $n\geq 6$ vertices into a 4-connected biplane graph?

\section*{Acknowledgements}
A. G., F. H., M. K., R.I. S. and J. T. were partially supported by ESF EUROCORES programme EuroGIGA, CRP ComPoSe: grant EUI-EURC-2011-4306, and  by project MINECO MTM2012-30951/FEDER. F. H., and R.I. S. were also supported by project Gen. Cat. DGR 2009SGR1040. A. G. and J. T. were also supported by project E58(ESF)-DGA. M.~K. was supported by the Secretary for Universities and Research of the Ministry of Economy and Knowledge of the Government of Catalonia and the European Union. I.~M. was supported by FEDER funds through COMPETE--Operational Programme Factors of Competitiveness, CIDMA and FCT within project PEst-C/MAT/UI4106/2011 with COMPETE number FCOMP-01-0124-FEDER-022690. M.~S.\ was supported by the
project NEXLIZ - CZ.1.07/2.3.00/30.0038, which is co-financed by
the European Social Fund and the state budget of the Czech
Republic, and by ESF EuroGIGA project ComPoSe as F.R.S.-FNRS - EUROGIGA NR 13604.
R.~S. was funded by Portuguese funds through CIDMA (Center for Research and Development in Mathematics and Applications) and FCT (Funda\c{c}\~{a}o para a Ci\^{e}ncia e a Tecnologia), within project PEst-OE/MAT/UI4106/2014, and by FCT grant SFRH/BPD/88455/2012.
C.~T. was supported in part by NSERC (RGPIN 35586) and NSF (CCF-0830734).

\end{document}